\documentclass[a4paper, 11pt]{article}
\usepackage{bbm}
\usepackage{verbatim}
\usepackage{amsfonts, amsmath, amssymb, amsthm}
\usepackage[english]{babel}
\usepackage[utf8x]{inputenc}
\usepackage[T1]{fontenc}

\usepackage{booktabs}
\usepackage{enumitem}
\usepackage{color, colortbl}
\usepackage{tabularx, makecell}

\usepackage{mathtools}
\usepackage{mathrsfs}
\usepackage{subcaption}
\usepackage[many]{tcolorbox}
\tcbuselibrary{listings,theorems}

\usepackage{tikz}
\usetikzlibrary{shapes,arrows,positioning,decorations.pathreplacing}
\usepackage[top=25mm,right=25mm,left=25mm,bottom=25mm,columnsep=20pt]{geometry}

\providecommand{\Sgn}{\operatorname{sgn}}                    
\providecommand{\supp}{\Supp}

\providecommand{\argmin}{\operatorname*{argmin}}  

\newcommand{\Va}{{\mathbf{a}}}
\newcommand{\Vb}{{\mathbf{b}}}
\newcommand{\Vc}{{\mathbf{c}}}
\newcommand{\Vd}{{\mathbf{d}}}
\newcommand{\Ve}{{\mathbf{e}}}

\newcommand{\Vm}{{\mathbf{m}}}

\newcommand{\Vu}{{\mathbf{u}}}
\newcommand{\Vv}{{\mathbf{v}}}
\newcommand{\Vw}{{\mathbf{w}}}
\newcommand{\Vx}{{\mathbf{x}}}
\newcommand{\Vy}{{\mathbf{y}}}
\newcommand{\Vz}{{\mathbf{z}}}

\providecommand{\Bb}{{\boldsymbol{b}}}

\newcommand{\VA}{{\mathbf{A}}}
\newcommand{\VB}{{\mathbf{B}}}
\newcommand{\VC}{{\mathbf{C}}}

\newcommand{\VI}{{\mathbf{I}}}

\newcommand{\VM}{{\mathbf{M}}}

\newcommand{\VP}{{\mathbf{P}}}
\newcommand{\VQ}{{\mathbf{Q}}}

\newcommand{\VU}{{\mathbf{U}}}
\newcommand{\VV}{{\mathbf{V}}}
\newcommand{\VW}{{\mathbf{W}}}

\newcommand{\VSigma}  {\mathbf{\Sigma}}

\newcommand{\Veta}       {\boldsymbol{\upeta}} %

\newcommand{\Vxi}        {\boldsymbol{\upxi}}

\newcommand{\Amat}{{\rm A}}

\newcommand{\Mmat}{{\rm M}}

\newcommand{\Umat}{{\rm U}}
\newcommand{\Vmat}{{\rm V}}
\newcommand{\Wmat}{{\rm W}}

\newcommand{\usf}{\mathsf{u}}

\newcommand{\wsf}{\mathsf{w}}

\newcommand{\abs}{\mathsf{a}}

\providecommand{\CE}{{\cal E}}
\providecommand{\CF}{{\cal F}}

\providecommand{\CN}{{\cal N}}
\providecommand{\CO}{{\cal O}}

\providecommand{\CT}{{\cal T}}

\providecommand{\bbN}{\mathbb{N}}

\providecommand{\bbR}{\mathbb{R}}

\providecommand{\Fb}{\mathfrak{b}}

\newcommand*{\EUSP}[2]{\left<{#1},{#2}\right>} 

\providecommand*{\N}[1]{\left\|{#1}\right\|} 
\newcommand*{\SN}[1]{\left|{#1}\right|}      
\providecommand*{\abs}[1]{\left|{#1}\right|} 

\newcommand*{\LRP}[1]{\left(#1\right)}

\theoremstyle{plain}
	\newtheorem{theorem}{\sffamily Theorem}[section]
	
	\newtheorem{lemma}[theorem]{\sffamily Lemma}

	\newtheorem{remark}[theorem]{\sffamily Remark}
	\newtheorem{definition}[theorem]{\sffamily Definition}

\allowdisplaybreaks

\allowdisplaybreaks

\usepackage[linkcolor=blue,colorlinks=true]{hyperref}

\newcommand{\BLUE}[1]{#1}

\renewcommand{\Bb}{\VB_\beta}
\renewcommand{\Amat}{\VA}
\renewcommand{\Mmat}{\VM}
\renewcommand{\Umat}{\VU}
\renewcommand{\Vmat}{\VV}
\renewcommand{\Wmat}{\VW}
\newcommand{\BbI}{\VB_{\beta,I}}
\newcommand{\uabq}{{\Vu_{\alpha,\beta}^q}}
\newcommand{\vabq}{{\Vv_{\alpha,\beta}^q}}
\renewcommand{\Vxi}{\boldsymbol{\mathrm{\xi}}}
\renewcommand{\Veta}{\boldsymbol{\mathrm{\eta}}}
\providecommand{\Id}{\mathsf{\VI{\Vd}}}
\newcommand{\umin}{{d_{\min}}}
\newcommand{\Tabq}{\mathcal{T}_{\alpha,\beta}^q}
\renewcommand{\Fb}{\mathcal{F}_\beta}
\newcommand{\indt}{\perp \!\!\! \perp }

\newcommand{\lmin}{\lambda_\text{min}}
\newcommand{\R}{\mathbb{R}}
\renewcommand{\supp}{\mathrm{supp}}
\newcommand{\sign}{\mathrm{sign}}

\newcommand{\round}[1]{\left( #1 \right)}

\newcommand{\curly}[1]{\left\{ #1 \right\}}
\renewcommand{\abs}[1]{\left| #1 \right|}

\newcommand{\inner}[1]{\left\langle #1 \right\rangle}

\providecommand*{\Nq}[1]{\|{#1}\|_{q}^q}         %
\providecommand*{\Ntwo}[1]{\|#1\|_{2}}         %
\providecommand*{\opN}[1]{\|#1\|}
\providecommand*{\prox}{\operatorname{prox}}     %

\tcbuselibrary{breakable}
\usepackage{authblk}
\usepackage[normalem]{ulem}
\title{Computational approaches to non-convex, sparsity-inducing multi-penalty regularization}
\author[1,3]{\v Zeljko Kereta\thanks{Email: \texttt{zeljko@simula.no}}}
\author[2]{Johannes Maly\thanks{Email: \texttt{johannes.maly@ku.de}}}
\author[3]{Valeriya Naumova\thanks{Email: \texttt{valeriya@simula.no}}}

\affil[1]{University College London, United Kingdom}
\affil[2]{KU Eichstaett/Ingolstadt, Germany}
\affil[3]{Simula Research Laboratory, Simula Metropolitan Center for Digital Engineering, Norway}
\date{}

\newcolumntype{C}[1]{>{\centering\arraybackslash}p{#1}}

\begin{document}
\maketitle

\begin{abstract}
  In this work we consider numerical efficiency and convergence rates for solvers of non-convex multi-penalty formulations when reconstructing sparse signals from noisy linear measurements. 
  We extend an existing approach, based on reduction to an augmented single-penalty formulation, to the non-convex setting and discuss its computational intractability in large-scale applications. 
  To circumvent this limitation, we propose an alternative single-penalty reduction based on infimal convolution that shares the benefits of the augmented approach but is computationally less dependent on the problem size. 
  We provide linear convergence rates for both approaches, and their dependence on design parameters. Numerical experiments substantiate our theoretical findings.
\end{abstract}


\section{Introduction}

In many real-life applications one is interested in recovering a structured signal from few corrupted linear measurements. 
One particular challenge lies in separating the ground-truth from pre-measurement noise since any such corruption is amplified during the measurement process, a phenomenon known as noise folding \cite{ACE11} or input noise model \cite{ASZ10}. 
It commonly appears in signal processing and compressed sensing applications, where noise is added to the signal both before and after the measurement process occurs. 
This can be modeled as
\begin{equation}\label{eqn:UM_IP}\Amat(\Vu^\dagger+\Vv) + \Vxi = \Vy,\end{equation}
where $\Vu^\dagger\in\bbR^n$ is an $s$-sparse original signal that we want to recover, $\Vv\in\bbR^n$ is the pre-measurement noise, $\Vxi \in \R^m$ is the post-measurement noise, and $\Amat\in\bbR^{m\times n}$ is the measurement matrix. \BLUE{Note that a signal $\Vu \in \mathbb{R}^n$ is called $s$-sparse if its support consists of at most $s$ elements, i.e. $|\supp(\Vu)| = |\curly{ i \colon \usf_i \neq 0 }| \le s$.}
Information theoretic bounds state that the number of measurements $m$ required for the exact support recovery of $\Vu^\dagger$ from \eqref{eqn:UM_IP} needs to scale linearly\footnote{Assume for simplicity $\Vv{\indt}\Vxi$, $\Vxi\sim\CN(0,\sigma^2\Id_n)$, and $\Vv\sim\CN(0, \sigma_v^2\Id_n)$. We now write \eqref{eqn:UM_IP} as $\Vy=\VA \Vu^\dagger + \Vw$, where $\Vw:=\VA \Vv+\Vxi$ represents the effective noise. 
The covariance matrix of $\Vw$ equals $\sigma^2\Id_m +\sigma_v^2\VA\VA^\top =: \VQ$. Assuming $\VA\VA^\top\approx \frac{n}{m}\Id_m$ (as is the case, with high probability, for $\VA$ with zero mean, $1/m$-variance sub-Gaussian entries), and $\sigma_v\approx \sigma$, we would have $\VQ = \sigma^2(1+C\frac{n}{m})\Id_m$, for $C>0$. 
Thus, the variance of the noise rises by a factor proportional to $n/m$, which when $m\ll n$ can be substantial.} with $n$, which leads to poor compression performance \cite{ASZ10}.

A number of recent studies \cite{PAF15, NP14,GN16, GKN18} try and mitigate these issues through a multi-penalty regularization framework defined as
\begin{equation}\label{eqn:mp_problem_gen} 
\min_{\Vu,\Vv\in\bbR^n} \frac{1}{2}\N{\Amat (\Vu+\Vv) -\Vy}_2^2 + \frac{\alpha}{q} \N{\Vu}_q^q + \frac{\beta}{p}\N{\Vv}_p^p,
\end{equation}
where $\alpha,\beta>0$ are regularization parameters, $0\leq q <2$, and $2\leq p < \infty$. 
In particular, to promote sparsity of the $\Vu$ component we choose $q\leq1$.
A natural  way to  minimize \eqref{eqn:mp_problem_gen} 
is via alternating minimization, starting from $\Vu^{0}, \Vv^{0} \in \R^n$ and then iterating as
\begin{align} \label{eqn:AM}
\begin{split}
    \Vu^{k+1} &\in \argmin_{\Vu\in\bbR^n} \frac{1}{2} \Ntwo{\Amat (\Vu+\Vv^k) - \Vy}^2 + \frac{\alpha}{q} \N{\Vu}_q^q, \\
    \Vv^{k+1} &\in \argmin_{\Vv\in\bbR^n} \frac{1}{2} \Ntwo{\Amat (\Vu^{k+1}+\Vv) - \Vy}^2 + \frac{\beta}{p} \N{\Vv}_p^p.
\end{split}
\end{align}
Whereas the second problem is differentiable and admits an explicit solution, the first problem requires iterative thresholding for $q\leq1$ \cite{NP14}, for each outer iteration $k\in\bbN$, and becomes non-convex if $q<1$. 
Moreover, alternating minimization does not lend itself to an easy analysis of the convergence rate.

\subsection{Contribution}

In this work we examine the multi-penalty problem \eqref{eqn:mp_problem_gen}, for the case $0<q\leq1$ and $p=2$. 
We first show that the augmented approach in \cite{GN16}, which allows to decouple the computation of $\Vu$ and $\Vv$ components of the solution, can be easily extended to $q<1$ to obtain an augmented single-penalty iterative thresholding algorithm providing solutions to \eqref{eqn:mp_problem_gen}. 
Since this includes computing the inverse of a possibly high-dimensional matrix, we suggest an alternative single-penalty iterative thresholding algorithm which is based on an infimal convolution formulation of \eqref{eqn:mp_problem_gen} and sidesteps the computational bottleneck of the augmented approach. 
We show a linear convergence rate for both approaches, in dependence of design parameters, and in numerical simulations confirm both the rate analysis and the efficiency gap. 
In particular, we argue that the benefits of faster convergence rates are sometimes offset by the computational demands, which suggests that a preferred method for solving the optimization problem can be chosen with respect to the size of $\VA$.

\subsection{Related Work}

In \cite{NP14} the authors approach \eqref{eqn:AM}, for $0 < q \le 1$ and $p= \infty$, on separable Hilbert spaces by applying iterative thresholding algorithms to each of the sub-problems, and show convergence of the sequence of iterates to stationary points of the underlying problem. \BLUE{The choice $p = \infty$ is of special interest when $\Vv$ models uniform pre-measurement noise. However, the authors also show that $p=2$ exhibits the best (empirical) performance for the reconstruction of $\Vu^\dagger$, for $\Vv$ modelling various common noise types (including uniform noise). It is for this reason that  in this paper we are concerned only with the case $p=2$.
We add though that more general noise types might be of interest in very particular cases, and this is a possible topic for future research.} 
In \cite{GN16} the authors reduce the optimization problem  \eqref{eqn:mp_problem_gen} to a single-penalty regularization through an augmented data matrix, for $q=1$ and $p=2$, and derive conditions on optimal support recovery. 
The authors provide theoretical and numerical evidence of superior performance of multi-penalty regularization over standard single-penalty  approaches for the sparse recovery of solutions to \eqref{eqn:UM_IP}. 
In \cite{GKN18} a principled, data-driven parameter selection approach is derived  for $q=1$ and $p=2$, based on the Lasso path.
Instead of through noise folding, a multi-penalty formulation of the objective function can also be seen from the perspective of the recovery of a signal that is a superposition of two components, e.g.\ a sparse and a smooth component. See \cite{DDD16} and references therein.
In spite of these and other advances, rigorous results regarding convergence rate and error analysis for \eqref{eqn:mp_problem_gen} have not been established.

Since we reduce \eqref{eqn:mp_problem_gen} to specific single-penalty problems, corresponding convergence results on classical proximal descent methods are of interest. 
In \cite{BL15} important insights on support stability and convergence of iterative thresholding algorithms on separable Hilbert spaces have been collected while \cite{ZLX14} proved linear convergence rates of the iterative thresholding algorithm, under certain conditions, if the underlying thresholding operator is not continuous, though the dependency on the parameters of the optimization scheme are not explicitly derived.
Linear convergence of a single penalty non-convex regularizer with adaptive thresholding was established in \cite{WZ15}, where the influence of the RIP of the design matrix on the convergence constant can be inferred.
A further survey of nonconvex regularizers for sparse recovery can be found in \cite{WC+18}.

Lastly, approaches representing regularizers as infimal convolution can be found in the context of machine learning and signal processing, cf.\  \cite{laude2018nonconvex,laude2019optimization}. 
Therein primal-dual schemes are examined for optimizing functionals penalized via infimal convolutions. The results, however, require piece-wise convexity which is not given in our case.

\subsection{Notation}
We restrict boldface lettering to matrices (uppercase), e.g.\ $\VA$, and vectors (lowercase), e.g.\ $\Vu$.
The $i^{\text{th}}$ entry of a vector $\Vu$ is denoted as $\usf_i$.
For $m\in\mathbb N$ we denote $[m]:=\{1,\ldots,m\}$.
For $0<q\leq\infty$ the $\ell_q$ norm of a vector $\Vu=(\usf_1,\ldots,\usf_n)^\top\in\bbR^n$ is denoted by $\N{\Vu}_q$.
The support set of  $\Vu\in\bbR^n$ is denoted as
\[ \supp(\Vu) = \{i \in [n] : \usf_i\neq 0\}\]
and the sign $\Sgn(\Vu)=(\Sgn(\usf_i))_{i=1}^n$ is defined component-wise by \[\Sgn(\usf)=
\begin{cases} 1,&\quad \textrm{if } \usf>0,\\
 0,&\quad \textrm{if } \usf=0,\\-1,&\quad \textrm{if }
\usf<0.\end{cases}
\]
\BLUE{For a matrix $\Mmat\in\bbR^{m\times n}$, we use $\N{\Mmat}$ to denote its spectral norm and $\lambda_{\min}(\Mmat)$ to denote its smallest singular value}. We denote the $n\times n$ identity matrix by $\Id_n$. 
For $I\subset[n]$, $\Mmat_I\in\bbR^{m\times\SN{I}}$ represents the submatrix of $\Mmat$ containing the columns indexed by $I$, and $\Vu_I \in \R^{\abs{I}}$ denotes the subvector of $\Vu$ containing the entries restricted to $I$. 
We denote the corresponding orthogonal projection operator onto $I$ as $\VP_I \in \R^{\abs{I}\times n}$, so that $\VP_I \Vu = \Vu_I$. When indexed by a set $T \subset \R^n$, $\VP_T$ denotes the orthogonal projection onto $T$.
\BLUE{Finally, the set-valued operator $\partial$ denotes the limiting Fr\'echet subdifferential, and $\operatorname{dom} \partial f = \curly{\Vx \colon \partial f (\Vx) \neq \emptyset}$ is its corresponding domain when applied to a function $f \colon \mathbb{R}^n \rightarrow \mathbb{R} \cup \curly{\infty}$, cf.\ \cite{rockafellar2009variational,mordukhovich2006variational}}.


\section{Main Results} 

Consider the multi-penalty problem \eqref{eqn:mp_problem_gen} for $p=2$, i.e.\ minimizing
\begin{equation}\label{eqn:mp_problem} 
 \CT_{\alpha,\beta}^q(\Vu,\Vv) := \frac{1}{2}\N{\Amat (\Vu+\Vv) -\Vy}_2^2 + \frac{\alpha}{q} \N{\Vu}_q^q + \frac{\beta}{2}\Ntwo{\Vv}^2,
\end{equation}
and denote a corresponding solution pair by
\begin{equation}
\label{eqn:mp_argmin}
    \LRP{\uabq,\vabq} \in \argmin_{\Vu,\Vv\in\bbR^n} \CT_{\alpha,\beta}^q(\Vu,\Vv). 
\end{equation}
As mentioned above $\VA\in\bbR^{m\times n}$, $\Vy\in\bbR^m$, $\alpha,\beta>0$ are regularization parameters balancing the contributions of the data-fidelity term and the two regularization terms, and $0<q\leq 1$.

Let us introduce two widely known concepts relevant for the forthcoming discussion. 
First, the \emph{Kurdyka-\L{}ojasiewicz} (K\L{}) property;  a well-established tool for analyzing the convergence, and convergence rates, of proximal descent algorithms \cite{attouch2010proximal}.

\begin{definition}
    A function $f:\bbR^n\rightarrow \bbR\cup\{\infty\}$ is said to have the {K\L{}} property at $\bar{\Vx} \in\operatorname{dom}\partial f$ if there exists $\eta\in(0,+\infty]$, a neighbourhood $\Omega$ of $\Vx$, and a continuous concave function $\varphi:[0,\eta)\to \bbR_+$ such that
    \begin{enumerate}
        \item $\varphi\in C^1\LRP{0,\eta}$, $\varphi(0)=0$ and $\varphi'(s)>0$ for all $s\in(0,\eta)$
        \item For all $\Vx\in\Omega\cap \{\Vx: f(\bar \Vx) < f(\Vx) < f(\bar \Vx) +\eta\}$ the K\L{} inequality holds
        \[ \varphi'\LRP{f(\Vx) - f(\bar \Vx)} \operatorname{dist}\LRP{0,\partial f(\Vx)} \geq 1.\]
    \end{enumerate}
\end{definition}

   The K\L{} property is used to describe the speed of convergence through the desingularizing function $\varphi$. 
   It has been  shown that semi-algebraic functions satisfy the K\L{} property with $\varphi(s)=cs^{1-\theta}$, where $c>0$ and $\theta\in[0,1)$ is called the K\L{} constant, which characterizes the convergence speed of proximal gradient descent algorithms \cite[Theorem 11]{attouch2010proximal}. 
   As observed in \cite{bolte2017error}, Corollary 3.6 in \cite{Li13} may be used to determine the K\L{} constant of piecewise convex polynomials.
   Even though $\Nq{\cdot}$ has the K\L{} property, cf.\ \cite[Example 5.4]{Attouch2013}, it does not result in piece-wise convex polynomials for $0<q<1$, and thus we cannot apply \cite[Corollary 3.6]{Li13} to infer the speed of convergence.
   We will instead adopt and adapt the ideas from \cite{BL15, ZLX14}.

The second concept relevant for this paper is the \emph{restricted isometry property (RIP)}, which allows to control eigenvalues of small submatrices of $\Amat \in \R^{m\times n}$, and to characterize measurement operators that allow stable and robust reconstruction of sparse signals from $m \ll n$ measurements.

\begin{definition}
    A matrix $\Amat \in \R^{m\times n}$ satisfies the restricted isometry property of order $s$ ($s$-RIP) with constant $\delta_s \in (0,1)$, if for all $s$-sparse $\Vu \in \R^n$
    \begin{align*}
        (1-\delta_s) \Ntwo{\Vu} \le \Ntwo{\Amat \Vu} \le (1+\delta_s) \Ntwo{\Vu}.
    \end{align*}
\end{definition}    

\begin{remark}\label{rem:RIP}
    For a detailed treatment of RIP, and measurement operators that fulfill it, we refer the reader to \cite{foucart:2013}. 
    Let us only mention that if the entries of $\Amat$ are i.i.d. copies of a Gaussian random variable with mean zero and variance $\frac{1}{m}$, then
    \begin{align*}
        m \ge C\delta_s^{-2} s \log\left( \frac{en}{s} \right)
    \end{align*}
    measurements suffice to have an $s$-RIP with constant $\delta_s> 0$ with high probability, for an absolute constant $C > 0$. Consequently, $\delta_s = \CO\big( m^{-1/2}\sqrt{s \log(en/s)} \big)$ with high probability.
\end{remark}


\subsection{Augmented Formulation} 

It was observed in  \cite{GN16} that for $q = 1$, the multi-penalty problem \eqref{eqn:mp_problem_gen} reduces to single-penalty regularization where measurement matrix and datum are adjusted by the regularization parameter $\beta$. 
We include this result, extended to $0 < q \le 1$, together with the proof (see Section \ref{app:mp_decoupled_sol}), which is analogous to \cite[Lemma 1]{GN16}.
\begin{lemma}\label{lem:mp_decoupled_sol}
The pair $(\uabq,\vabq)$ minimizes $\CT_{\alpha,\beta}^q$ in \eqref{eqn:mp_problem} if and only if 
\begin{align} \label{eqn:vOFu}
    \vabq = v(\uabq) = \LRP{\beta\,\Id_n + \Amat^\top\Amat}^{-1}\LRP{\Amat^\top \Vy-\Amat^\top\Amat \uabq},
\end{align}
and  $\uabq$ is the solution of the \emph{augmented problem}
\begin{equation}\label{eqn:q_augmented_problem} 
    \uabq \in \argmin_{\Vu\in\bbR^n} \CF_\beta(\Vu), \quad \CF_\beta(\Vu) := \frac{1}{2}\N{\Bb \Vu - \Vy_\beta}_2^2 + \frac{\alpha}{q}\N{\Vu}_q^q,
\end{equation}
with
\begin{align*}
    \Bb = \LRP{\Id_m+\frac{\Amat\Amat^\top}{\beta}}^{-1/2}\Amat \quad \text{and} \quad \Vy_\beta =\LRP{\Id_m+\frac{\Amat\Amat^\top}{\beta}}^{-1/2} \Vy.
\end{align*}{}
\end{lemma}

\begin{remark}
The noise folding forward model \eqref{eqn:UM_IP} is in \cite{ACE11} written in the whitened form as $\tilde\Vy = \VB \Vu^\dagger+\Veta $, for $\tilde\Vy = \VQ^{-1/2}\Vy$, $\VB=\VQ^{-1/2}\Amat$, $\Veta = \VQ^{-1/2}(\Amat\Vv+\Vxi)$, for $\VQ=\frac{1}{c}(\sigma^2\Id_m+\sigma_v^2\Amat\Amat^\top)$ and $c>0$ is a constant.    
Notice that this is particularly related to the augmented problem in \eqref{eqn:q_augmented_problem}.\\
On an unrelated note, improving on the analysis in \cite[Proposition 2]{ACE11} one can show (see Lemma \ref{lem:BbCoherence}) that the coherence, defined for a matrix $\VM$ as
\[
    \mathrm{coh}(\VM) = \max_{i\neq j} \frac{\SN{\Vm_i^\top\Vm_j}}{\Ntwo{\Vm_i}\Ntwo{\Vm_j}},
\]
where $\Vm_i$ is the $i$-th column of $\VM$, of the augmented measurement matrix $\Bb$ satisfies
\begin{align} \label{eq:CohBound}
    \mathrm{coh}(\Bb)\leq\LRP{1+\frac{\opN{\Amat}^2}{\beta}}\LRP{\mathrm{coh}(\VA)+\frac{\opN{\Amat}^2}{\beta}}.
\end{align}
\BLUE{In compressed sensing literature, the magnitude of the coherence of a matrix is an important measure of quality for measurement matrices, cf.\ \cite[Section 5]{foucart:2013}. The bound in \eqref{eq:CohBound} thus suggests that for small $\N{\Amat}$ or large $\beta$, the linear measurement process modelled by $\Bb$ is as information preserving as the one modelled by $\Amat$.}
In addition, Lemma \ref{lem:BbCoherence2} shows that $\mathrm{coh}(\Bb)$ behaves like the coherence of a conditioned version of $\Amat$ if $\beta \rightarrow 0$.
Let us mention that in practice $\mathrm{coh}(\Bb)$ behaves well for all $\beta$'s, and even moderate values of $\opN{\Amat\Amat^\top}$.
\end{remark}
By Lemma \ref{lem:mp_decoupled_sol}, to estimate the solution pair $(\uabq,\vabq)$ it is sufficient to first solve \eqref{eqn:q_augmented_problem}, and then insert the computed solution into \eqref{eqn:vOFu}.
Since the fidelity term $\frac{1}{2}\Ntwo{\Bb \Vu -\Vy_\beta}^2$ is smooth and the regularization term $\N{\Vu}_q^q$ non-convex, the common approach
is to use iterative thresholding through a forward-backward splitting algorithm \cite{BL15, attouch2010proximal}. 
For $\Fb$  and the augmented problem \eqref{eqn:q_augmented_problem}, the resulting thresholding iterations applied  are readily written as
\begin{equation} \label{eqn:algorithm}
    \begin{cases}\text{Set the initial vector } \Vu^0  \\
    \Vu^{k+1} = \prox_{\mu, \frac{\alpha}{q} \Nq{\cdot}} (\Vu^k - \mu \Bb^\top(\Bb \Vu^k - \Vy_\beta)) . \end{cases}
\end{equation}
Each iteration in \eqref{eqn:algorithm} can be viewed as a thresholded Landweber iteration; we first perform a step in the direction of the negative gradient of the data fidelity term, and then apply the proximal operator of the remaining non-convex term.

The proximal operator of a function $\Psi:\bbR^n\rightarrow\bbR^n$ is defined by
\begin{align} \label{eqn:proxOp}
    \prox_{\mu,\nu\Psi}(\Vu) =  \argmin_{\Vz\in\bbR^n} \frac{1}{2\mu} \Ntwo{\Vz-\Vu}^2 +\nu\Psi(\Vz),
\end{align}
where $\mu,\nu>0$. 
For separable mappings
\eqref{eqn:proxOp} can be applied component-wise, and we have
$
 \prox_{\mu,\nu\Nq{\cdot}}(\Vu) = 
 \LRP{\prox_{\mu,\nu|\cdot|^q}(\usf_i)}_{i=1}^n.$
In the general case, the proximal operator \eqref{eqn:proxOp} could be set-valued, since there might be multiple or even no minima.
It can be shown though that for $0<q<1$ the (one-dimensional) proximal operator of $\abs{\cdot}^q$ satisfies
\begin{align}\label{eqn:prox_lp}
    \prox_{\mu,\nu|\cdot|^q }(\usf) &= \begin{cases}\LRP{\cdot+\nu\mu q \Sgn({\cdot})\SN{\cdot}^{q-1}}^{-1} (\usf), &\text{for } \SN{\usf} > \tau_\mu\\ 0, &\text{for } \SN{\usf}\leq \tau_\mu\end{cases}, \\
    \text{where } \tau_\mu &=\frac{2-q}{2-2q}\LRP{2\nu\mu(1-q)}^{\frac{1}{2-q}}\nonumber.
\end{align}
The range of $\prox_{\mu,\nu\SN{\cdot}^q }$ is $(-\infty,-\lambda_{\mu,q}] \cup \{0\}\cup [\lambda_{\mu,q},\infty)$ where $\lambda_{\mu,q} = \LRP{2\mu\nu(1-q)}^{\frac{1}{2-q}}$, see \cite[Lemma 5.1]{BL15}, and it is discontinuous with a jump discontinuity\footnote{While the actual proximal operator of $\abs{\cdot}^q$ is set-valued and simultaneously assumes both possible values at $\abs{\usf} = \tau_\mu$, we follow common practice when restricting the operator to zero at $\abs{\usf} = \tau_\mu$ to have a single-valued function.} at $|\usf| = \tau_\mu$.
Note that the proximal operators in \eqref{eqn:prox_lp} are indeed thresholding operators, and as $q$ goes from $0$ to $1$ they interpolate between hard- and soft-thresholding operators.
Moreover, a closed form of the operator $\prox_{\mu,\nu|\cdot|^q }$ is known only in special cases, namely for $q=1/2$ and $q=2/3$ \cite{XC12}.

It follows easily that if the step-size $\mu>0$ is small enough (smaller than $\opN{\Bb}^{-2}$), the difference of iterates in \eqref{eqn:algorithm} decreases, i.e.\ $\Ntwo{\Vu^{k+1}-\Vu^k}\rightarrow 0$ as $k\rightarrow \infty$, see \cite[Proposition 2.1]{BL15}.
Note that the iterations in \eqref{eqn:algorithm} are quite different from those given by alternating minimization, where for each $k$ we need to compute $\Vu^{k+1}$ through iterative thresholding.
The following lemma makes this more precise; it shows that \eqref{eqn:algorithm} is equivalent to performing only the first step of iterative thresholding when computing $\Vu^{k+1}$ in \eqref{eqn:AM}. The proof can be found in Section \ref{app:AM}.

\begin{lemma} \label{lem:AM}
    The iterations defined in \eqref{eqn:algorithm} can be rewritten as
    \begin{align*}
        \Vu^{k+1} = \prox_{\mu, \frac{\alpha}{q} \Nq{\cdot}} (\Vu^k - \mu \Amat^\top(\Amat \Vu^k + \Amat v(\Vu^k) - \Vy)),
    \end{align*}
    which corresponds to a single proximal gradient descent step of \eqref{eqn:AM} starting at $\Vu^k$.
\end{lemma}
\subsubsection{Linear Convergence}

\BLUE{We now show that the iterates in \eqref{eqn:algorithm} converge at a linear rate to stationary points $\Vu^\star$ of $\CT_{\alpha,\beta}^q$, i.e. points such that $\boldsymbol{0} \in \partial \CT_{\alpha,\beta}^q(\Vu^\star)$, and characterize the convergence constant in dependence of design parameters.} 
Let us emphasize that since our analysis is tailored to $\ell_q$-regularization we derive more explicit guarantees (in terms of the involved parameters) than what would follow by directly applying the more general statements of \cite{ZLX14} to the augmented formulation \eqref{eqn:q_augmented_problem}.
The proof can be found in Section \ref{app:Rho}.
\begin{theorem} \label{thm:Rho}
    Let $\alpha, \beta > 0$ and $0<q\leq1$.
    Assume the matrix $\Amat\in\bbR^{m\times n}$ has  RIP of order $s$ with a constant $\delta_s \in (0,1)$, and let the stepsize $\mu$ satisfy $0 < \mu < \opN{\Amat}^{-2} + \beta^{-1}$.
    Moreover, assume\footnote{The sequence $\Vu^k$ converges provably to a stationary point since $\CT_{\alpha,\beta}^q$ is among other things coercive and has the KL-property, cf.\ \cite[Theorem 5.1]{Attouch2013}. The assumption thus is not about whether $\Vu_k$ converges but about the specific limit point which mainly depends on the concrete choice of initialization. \label{footnote:Convergence}}  $\Vu^\star\in\bbR^n$ is such that $|\supp(\Vu^\star)| \le s$ and the iterates \eqref{eqn:algorithm} satisfy $\Vu^k \rightarrow \Vu^\star$. 
    Define $I = \supp(\Vu^\star)$ and $\umin = \min_{i \in I} |\usf^\star_i|$. 
    Then there exists $k_0 \in \mathbb{N}$ such that for all $k \ge k_0$ we have
    \begin{align*}
        \Ntwo{\Vu^{k+1} - \Vu^\star}\leq\frac{1-\mu\left( 1 + \frac{\opN{\Amat}^2}{\beta} \right)^{-1} (1-\delta_s)^2}{ 1- \mu \alpha (1-q) \left( \frac{\umin}{2} \right)^{q-2} } \Ntwo{\Vu^{k} - \Vu^\star}.
    \end{align*}
\end{theorem}
\begin{remark} \label{rem:Rho}
\begin{enumerate}[label=(\roman*), itemsep=-5pt,partopsep=1ex,parsep=1ex, topsep=0pt]
    \item To have linear convergence in Theorem \ref{thm:Rho}, we have to choose $\alpha$ such that
    \begin{align} \label{eq:AlphaBound}
        0 < \alpha < \alpha^\star = \left( 1 + \frac{\opN{\Amat}^2}{\beta} \right)^{-1} \frac{(1-\delta_s)^2}{(1-q)}  \left( \frac{\umin}{2} \right)^{2-q}.
    \end{align}
This resembles basic assumptions of the main result in \cite{ZLX14}. One should thus interpret Theorem \ref{thm:Rho} as an additional refinement, better capable of predicting numerical behavior.
\item Theorem \ref{thm:Rho} suggests that the convergence constant depends on the sparsity of the signal and properties of $\Amat$. 
Namely, if the signal is sparser (and thus $\delta_s$ smaller) then the convergence constant decreases. Similarly, the constant decreases if we increase the number of measurements.
\item Assuming $\alpha = c\alpha^\star$, for $c \in (0,1)$, it is straight-forward to check that the rate in Theorem \ref{thm:Rho} becomes minimal by choosing $\mu \approx \opN{\Amat}^{-2} + \beta^{-1}$. In this case the result transforms into
    \begin{align*}
        \Ntwo{\Vu^{k+1} - \Vu^\star} \leq \frac{1 - \opN{\Amat}^{-2} (1-\delta_s)^2}{1 - c \opN{\Amat}^{-2} (1-\delta_s)^2} \Ntwo{\Vu^{k} - \Vu^\star}.
    \end{align*}
\item \BLUE{Since $\alpha$ and $\beta$ control the strength of regularization in $\CT_{\alpha,\beta}^q$, their choice depends on the expected noise level. Consequently, when setting $\alpha$ and $\beta$ one needs to make a trade-off between their regularizing effect and the desired convergence speed.}
\end{enumerate}
\end{remark}

\subsubsection{Computational Complexity}
\label{sec:AugCompComp}

Once $\Bb$ has been computed, executing \eqref{eqn:algorithm}  for a constant number of iterations costs $\CO(mn)$ operations: $\CO(mn)$ for matrix-vector products and $\CO(n)$ for evaluating the proximal operator.
But this gets dominated by the operations needed to obtain $\Bb$, which involve a matrix square root and a matrix-matrix linear system and have to be done in advance.
This turns out to be a computational bottleneck as soon as $m \ge n^{\frac{1}{\rho - 1}}$ as it requires $\CO(m^\rho)$ operations, where $\rho \in [2.37,3]$ depends on the used algorithmic method \cite{cormen2009introduction}. 
Such a computational cost can be prohibitive for high-dimensional applications.

\subsection{Infimal Convolution Formulation} 

To overcome the computational limitations observed above, we consider an alternative approach. 
Define a new program by
\begin{align} \label{eqn:mp_infconv}
    \Vw_{\alpha,\beta}^q = \argmin_{\Vw \in \R^n} \frac{1}{2} \Ntwo{\Amat \Vw - \Vy}^2 + \round{\frac{\alpha}{q} \Nq{\cdot} \Delta \frac{\beta}{2} \Ntwo{\cdot}^2} (\Vw),
\end{align}
where the infimal convolution is given by 
\begin{align} \label{eqn:InfConv}
    g(\Vw) := \round{\frac{\alpha}{q} \Nq{\cdot} \Delta \frac{\beta}{2} \Ntwo{\cdot}^2} (\Vw) = \inf_{\Vu \in \R^n} \frac{\alpha}{q} \Nq{\Vu} + \frac{\beta}{2} \Ntwo{\Vw - \Vu}^2.
\end{align}
For a detailed treatment of infimal convolution and its properties, see \cite{bauschke2011convex}. It is straight-forward to check that an equivalence between minimizing \eqref{eqn:mp_problem} and \eqref{eqn:mp_infconv} holds.
\begin{lemma}
    The pair $(\uabq,\vabq)$ minimizes $\CT_{\alpha,\beta}^q$ in \eqref{eqn:mp_problem} if and only if $\uabq + \vabq$ solves \eqref{eqn:mp_infconv} while $\uabq$ attains the infimal value of $\round{\frac{\alpha}{p} \Nq{\cdot} \Delta \frac{\beta}{2} \Ntwo{\cdot}^2} (\uabq + \vabq)$.
\end{lemma}

In order to solve \eqref{eqn:mp_infconv} via iterative thresholding (i.e.\ proximal gradient descent), we need to efficiently evaluate the proximal operator of \eqref{eqn:InfConv}. 
A helpful observation is that \eqref{eqn:InfConv} can be interpreted as the Moreau-envelope of $\Nq{\cdot}$, which for a function $f$ and $t>0$ is defined as
\begin{align*}
    M_{t,f}(\Vx) = \round{ f \Delta \frac{1}{2t} \Ntwo{\cdot}^2 } (\Vx) =  f(\prox_{t,f}(\Vx)) + \frac{1}{2t} \Ntwo{\Vx - \prox_{t,f}(\Vx)}^2,
\end{align*}
where the last equality only holds if $\prox_{t,f}(\Vx) \neq \emptyset$. It has been observed in \cite[Theorem 6.63]{beck2017first} that computing the proximal operator of the Moreau envelope reduces to computing the proximal operator of the underlying function. 
Though stated only for convex functions in \cite{beck2017first}, it is straight-forward to generalize the result. 
%
\begin{lemma} \label{lem:Moreau_prox}
    Let $f \colon \R^n \rightarrow \R$ be a lower semi-continuous function with $f(0) = \min f$. Then,
    \begin{align*}
        \prox_{\mu,\lambda M_{t,f}} (\Vx) = \frac{t}{t + \mu \lambda} \Vx + \frac{\mu \lambda}{t + \mu \lambda} \prox_{(t+\mu\lambda), f}(\Vx).
    \end{align*}
\end{lemma}
The proof is in Section \ref{app:Moreau_prox}.
Define now the proximal gradient descent for \eqref{eqn:mp_infconv} by
\begin{align}\label{eqn:InfConv_algorithm}
    \begin{cases}
    \text{Set the initial vector } \Vw^0  \\
    \Vw^{k+1} = \prox_{\mu,g} ( \Vw^k - \mu \Amat^\top (\Amat \Vw^k - \Vy)).
    \end{cases}
\end{align}
We denote  by $\Vu^k = \prox_{\frac{1}{\beta},\frac{\alpha}{q} \Nq{\cdot}} (\Vw^k)$ the sequence of minimizers attaining $g(\Vw^k)$, and set $\Vv^k = \Vw^k - \Vu^k$. Note that with this notation $\Vw^k$ and $\Vu^k$ can also be characterized via
\begin{align} \label{eqn:InfConv_AMrepresentation}
    \begin{cases}
    \Vw^k = \argmin_{\Vw \in \R^n} \frac{1}{2\mu} \Ntwo{\Vw - \Vw^{k-1} + \mu \Amat^\top (\Amat \Vw^{k-1} - \Vy)}^2 + \frac{\beta}{2} \Ntwo{\Vw - \Vu^k}^2 \\
    \Vu^k = \argmin_{\Vu \in \R^n} \frac{\beta}{2} \Ntwo{\Vu - \Vw^k}^2 +  \frac{\alpha}{q} \Nq{\Vu}
    \end{cases}.
\end{align}{}
\BLUE{Unlike \eqref{eqn:InfConv_algorithm}, the representation in \eqref{eqn:InfConv_AMrepresentation} does not yield a practically viable algorithm, since $\Vw^k$ and $\Vu^k$ are not decoupled. It does though lend itself to theoretical analysis of the iterations, cf. Section \ref{app:IC_Convergence}.}

\subsubsection{Linear Convergence}

Though $g$ in \eqref{eqn:InfConv} is continuous and separable, i.e. $g(\Vw) = \sum_{i=1}^n g_i(\wsf_i)$, it is not continuously differentiable, such that we cannot apply \cite{ZLX14} to deduce linear convergence of \eqref{eqn:InfConv_algorithm}. Nevertheless, using the KKT-conditions of the objective functions in \eqref{eqn:InfConv_AMrepresentation}, we get linear convergence of the iterates in \eqref{eqn:InfConv_algorithm} by a similar strategy as in Theorem \ref{thm:Rho}.
\begin{theorem} \label{thm:IC_Convergence}
    Let $\alpha, \beta > 0$ and $0<q\leq1$. 
    Assume\footnote{Along the lines of Footnote \ref{footnote:Convergence} in Theorem \ref{thm:Rho}. Just note that $g$ in \eqref{eqn:InfConv} has the KL-property by \cite[Theorem 3.1]{yu2019deducing} and, hence, the objective function in \eqref{eqn:mp_infconv} has it as well. \label{footnote:Convergence2}} that $0 < \mu < \opN{\Amat}^{-2}$ and $\Vw^k \rightarrow \Vw^\star$. 
    Let $I \subset [n]$ denote the support of $\Vu^\star = \prox_{\frac{1}{\beta},\frac{\alpha}{q} \Nq{\cdot}} (\Vw^\star)$ and define $\umin = \min_{i \in I} \abs{\usf_i^\star}$. 
    Then there exists $k_0 \in \mathbb{N}$ such that for all $k \ge k_0$ we have
    \begin{align*}
        \Ntwo{\Vw^{k+1} - \Vw^\star} \le \round{\frac{\opN{ \VP_I - \mu \Amat_I^\top \Amat }^2 }{ \Big({1 - \alpha\mu (1-q) \round{\frac{\umin}{2}}^{q-2} \Big)^2 }} + \frac{\opN{ \VP_{I^c} - \mu \Amat_{I^c}^\top \Amat }^2 }{(1 + \mu\beta)^2}}^{1/2} \; \Ntwo{\Vw^k - \Vw^\star}
    \end{align*}{}
\end{theorem}{}
The proof of Theorem \ref{thm:IC_Convergence} is given in Section \ref{app:IC_Convergence}.
\begin{remark}
On the one hand, in Theorem \ref{thm:IC_Convergence} the assumption on $\mu$ and the rate differ from Theorem \ref{thm:Rho}; there is no influence of $\beta$ on admissible step-sizes and the rate is split in two distinct components. 
On the other hand, since, for $\mu < \opN{\Amat}^{-2}$,
\begin{align} \label{eq:OperatorBound}
\begin{split}
    \opN{ \VP_I - \mu \Amat_I^\top \Amat } 
    &= \opN{ \VP_I (\Id_n - \mu \Amat^\top \Amat) }
    \le \opN{\Id_n - \mu \Amat^\top \Amat }
    < 1 \text{ and }\\
    \opN{ \VP_{I^c} - \mu \Amat_{I^c}^\top \Amat } 
    &= \opN{ \VP_{I^c} (\Id_n - \mu \Amat^\top \Amat) }
    \le \opN{\Id_n - \mu \Amat^\top \Amat }
    < 1,
\end{split}
\end{align} 
the rate in Theorem \ref{thm:IC_Convergence} suggests to choose $\beta$ large to dominate the second term of the rate in which case the assumptions on $\mu$ agree in both theorems. Moreover, this reduces the rate to 
\begin{align*}
        \Ntwo{\Vw^{k+1} - \Vw^\star} \le  \round{ \frac{\opN{ \VP_I - \mu \Amat_I^\top \Amat } }{ 1 - \alpha\mu (1-q) \round{\frac{\umin}{2}}^{q-2} } + \mathcal{O}(\beta^{-1}) } \; \Ntwo{\Vw^k - \Vw^\star},
\end{align*}{}
where the denominator is as in Theorem \ref{thm:Rho}. In light of \eqref{eq:OperatorBound}, we get linear convergence of \eqref{eqn:InfConv_algorithm} if
\begin{align*}
    0 < \alpha < \alpha^* = \frac{1 - \opN{ \VP_I - \mu \Amat_I^\top \Amat }}{\mu (1 - q) } \round{\frac{\umin}{2}}^{2-q}.
\end{align*}
\BLUE{As already discussed in Remark \ref{rem:Rho}, a trade-off between regularization and convergence rate has to be taken into account when choosing $\alpha$ and $\beta$.}
\end{remark}

\begin{remark} \label{rem:MoreauInterpretation}
    For $q = 1$, an alternative viewpoint on \eqref{eqn:InfConv_AMrepresentation} is given by
    \begin{align}\label{eqn:l1smooTH}
    \begin{split}
        \Vw^{k+1} &= \argmin_{\Vw \in \R^n} \frac{1}{2\mu} \Ntwo{\Vw - \Vw^k + \mu \Amat^\top (\Amat \Vw^k - \Vy)}^2 + \frac{\beta}{2} \Ntwo{\Vw - \Vu^{k+1}}^2\\
        &= \argmin_{\Vw \in \R^n} \frac{1}{2\mu} \Ntwo{\Vw - \Vw^k + \mu \Amat^\top (\Amat \Vw^k - \Vy)}^2 + \frac{\beta}{2} \Ntwo{\Vw - \prox_{\frac{\alpha}{\beta} \N{\cdot}_1}(\Vw)}^2\\
        &= \argmin_{\Vw \in \R^n} \frac{1}{2\mu} \Ntwo{\Vw - \Vw^k + \mu \Amat^\top (\Amat \Vw^k - \Vy)}^2 + \frac{\alpha}{2} \Ntwo{ \nabla M_{\frac{\alpha}{\beta} \N{\cdot}_1} (\Vw) }^2,
    \end{split}
    \end{align}
    where we used \cite[Eq. (3.3)]{parikh2014proximal} in the last step, meaning that
    \begin{align*}
        \Vw^{k+1} = \prox_{ \frac{\alpha \mu}{2} \Ntwo{ \nabla M_{\frac{\alpha}{\beta} \N{\cdot}_1} (\cdot) }^2} (\Vw^k - \mu \Amat^\top (\Amat \Vw^k - \Vy))
    \end{align*}
    is a proximal gradient descent sequence of $\Ntwo{ \nabla M_{\frac{\alpha}{\beta} \N{\cdot}_1} (\cdot) }^2$, the squared $\ell_2$-norm of the gradient of the smooth Moreau approximation of $\frac{\alpha}{\beta} \N{\cdot}_1$. 
    From this perspective, multi-penalty regularization resembles a Newton-type method by searching for zeros of the derivative of a smooth approximation of the $\ell_1$-norm. 
    However, transferring this intuition to the case $q<1$ is non-trivial. 
    On a technical level the equations in \eqref{eqn:l1smooTH} break down in the third line, which does not hold for $q<1$ due to non-convexity of $\Nq{\cdot}$.
\end{remark}

\subsubsection{Computational Complexity}

While \eqref{eqn:algorithm} requires computing $\Bb$, which can be costly, 
 the infimal convolution formulation \eqref{eqn:InfConv_algorithm} does not incur additional computational costs and thus directly inherits efficiency and linear convergence of the proximal descent method.
Indeed, for a fixed number of iterations the number of operations performed in \eqref{eqn:InfConv_algorithm} is $\CO(mn)$ (the additional convex combination when evaluating the proximal operator by Lemma \ref{lem:Moreau_prox} is negligible).
This is considerably lower than $\CO(m^\rho)$, for $\rho \in [2.37,3]$, which is the computational cost of the augmented formulation, particularly if $m$ is large. In numerical simulations, this effect is easy to observe, cf. Section \ref{sec:Numerics}.


\section{Numerical Experiments}
\label{sec:Numerics}

We now present experimental results that focus on two aspects of our study.
First, we examine the convergence rate of the proposed algorithms, confirming linear convergence and in case of the augmented formulation, the dependence of the convergence constant on the parameters of the problem.
Second, we examine their efficiency by studying the overall computational effort on larger scale problems.


\subsection{Convergence Rate}
\label{sec:Experiment1}

\begin{figure}[!ht]
\scriptsize
\centering
\begin{subfigure}[c]{0.48\textwidth}
    \includegraphics[width=\textwidth]{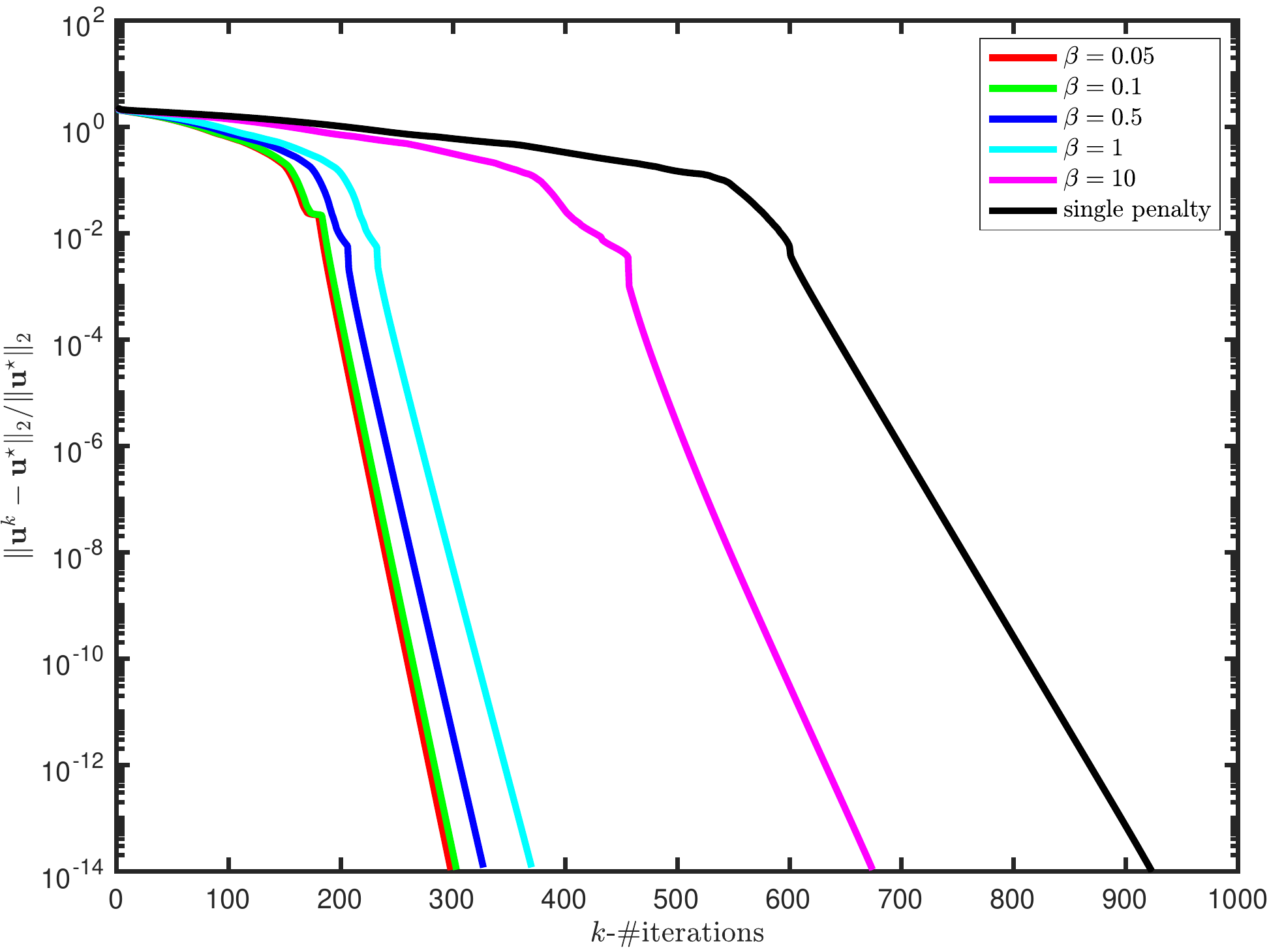}
    \caption{Varying $\beta$.}
    \label{fig:ComparisonA}
\end{subfigure} 
\quad
\begin{subfigure}[c]{0.48\textwidth}
    \includegraphics[width=\textwidth]{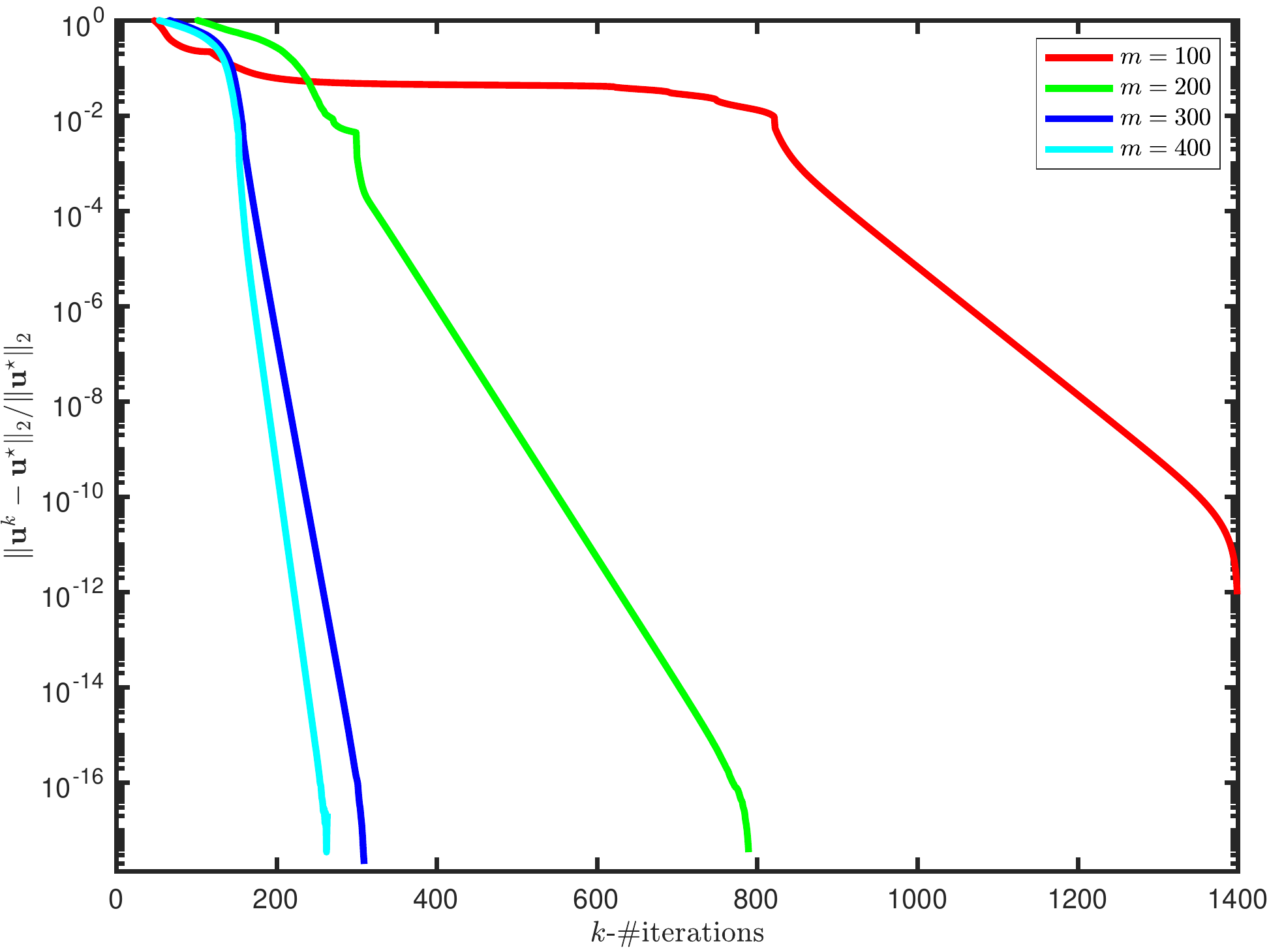}
    \caption{Varying $m$.}
    \label{fig:ComparisonB}
\end{subfigure}
    \caption{In the left panel we consider $\Amat\in\bbR^{200\times600}$ and vary the parameter $\beta$, whereas in the right panel we consider $\Amat\in\bbR^{m\times600}$ and vary the number of measurements $m\in\{100,200,300,400\}$.}
    \label{fig:Comparison}
\end{figure}
Via the RIP-constant $\delta_s$ Theorem \ref{thm:Rho} gives a direct dependence of the convergence rate on the sparsity of the solution and the properties of the matrix, whereas Theorem \ref{thm:IC_Convergence} is harder to interpret: it is straight-forward to deduce the existence of parameter regimes in which linear convergence occurs but hard to quantify the rate in terms of the parameters.
While numerical evidence for linear convergence of the infimal convolution formulation is observed in Section \ref{sec:Experiment2}, we continue by validating Theorem \ref{thm:Rho} in two experiments. In both, we take $q=1/2$, and add pre- and post-measurement Gaussian noise terms, $\Vv$ and $\Vxi$, with \BLUE{noise level $\frac{\|{\Vv}\|}{\|{\Vu^\dagger}\|}=\frac{\|{\Vxi}\|}{\|{\Vu^\dagger}\|}=0.1$}. 
We choose an admissible $\alpha$ according to Remark \ref{rem:Rho} and tune it such that the reconstructed signal shares its support size with the ground-truth. Both illustrations in Figure \ref{fig:Comparison} plot the relative error between the iterates $\Vu^k$ and the stationary point $\Vu^\star$ against the number of proximal gradient descent steps.

\paragraph{Varying the Penalty Parameter.} In the first experiment we take a Gaussian matrix $\VA\in\bbR^{200\times600}$, a $20$-sparse signal $\Vu^\dagger$, and vary $\beta$.
Theorem \ref{thm:Rho} predicts that smaller values of $\beta$ allow to take larger stepsizes, though the convergence constants are (essentially) the same.
This effect is readily observed in Figure \ref{fig:ComparisonA}.
Note that we can also observe that for smaller $\beta$ the algorithm reaches the steep part of the curve faster.
This is due to the fact that the convergence of iterates is initially slow (until the support is identified) and larger step-sizes allow to reduce the support size faster. 
The overall speed-up allowed by a smaller $\beta$ can be by up to a two-fold, in terms of the number of iterations needed to reach the  desired accuracy level.

\paragraph{Varying the Measurements.} In the second experiment we consider a Gaussian matrix $\VA\in\bbR^{m\times600}$, for $m\in\{100, 200, 300, 400\}$, and a $20$-sparse signal $\Vu^\dagger$.
Varying the number of measurements changes the RIP of the measurement matrix (a larger $m$ decreases $\delta_s$, see Remark \ref{rem:RIP}), and per Theorem \ref{thm:Rho} should affect the convergence constant. 
Figure \ref{fig:ComparisonB} shows exactly that. 
An analogous effect can be observed for different classes of measurement matrices, such as partial Toeplitz, or partial circulant matrices with Rademacher or Gaussian entries, but those results have not been included for the sake of brevity.

\subsection{Computational Comparison}
\label{sec:Experiment2}

\paragraph{Iteration Count.}
In order to provide numerical evidence for our initial statement that alternating minimization is highly sub-optimal, in Figure \ref{fig:CCComparisonA} we look at the decay of the relative error over the number of basic iterations, i.e.\ the number of thresholded gradient descent steps, of all three discussed approaches: alternating minimization \eqref{eqn:AM}, augmented formulation \eqref{eqn:algorithm}, and infimal convolution \eqref{eqn:InfConv_algorithm}. 
\BLUE{In this experiment, we use a Gaussian matrix $\Amat\in\bbR^{100\times500}$, the original signal is $14$-sparse, $q=1/2$ and the parameter $\alpha$, $\beta$, and $\mu$ are selected so that each method returns a $13$-sparse vector.} The $x$-axis refers to the number of times the proximal operator is called while the $y$-axis shows the relative error. The considerably worse performance of alternating minimization is due to the fact that it requires (too) many thresholded gradient steps to solve, for each $k\in\bbN$, sub-problems for the $\Vu^k$ component up to pre-fixed accuracy $\varepsilon = 10^{-8}$. 
Thus, the algorithm performs hardly any alternating steps. 

\paragraph{Computation Time.} 
To now illustrate the differences between augmented and infimal convolution formulation in terms of computational complexity, we perform the following experiment. We set the parameters generically to $\alpha = 0.02$, $\beta = 0.2$, and $\mu = 0.1$, and reconstruct a $100$-sparse signal $\Vu^\dagger \in \R^{5000}$ from measurements $\Vy \in \R^m$, for $m$ varying from $1000$ (sub-sampling) to $8000$ (over-sampling). We again take $q=1/2$, and add pre- and post-measurement noise terms, $\Vv$ and $\Vxi$, with noise level $0.1$. Averaging over $20$ random realizations of $\Vu^\dagger$, we record for augmented \eqref{eqn:algorithm} and infimal convolution approach \eqref{eqn:InfConv_algorithm} the time needed to perform $50$ iterations. 
After such few iterations none of the two algorithms has converged, though this already suffices to make a point regarding the computational cost since both algorithms incur the same cost (i.e.\ the gap remains the same) in the remaining iterations.
As Figure \ref{fig:CCComparisonB} shows, the additional computation of $\Bb$ in \eqref{eqn:algorithm} causes a massive additional workload leading to limited applicability of the augmented approach in large-scale settings. In contrast, the infimal convolution formulation is hardly affected by the increase in the number of measurements. Though the augmented approach tends to converge in fewer iterations, cf.\ Figure \ref{fig:CCComparisonA}, the additional iterations needed by the infimal convolution formulation to reach a comparable level of accuracy do not close the gap in computation time.
Note that we do not include alternating minimization here since it requires many more iterations (in the sense of single thresholded gradient descent steps) to show similar reconstruction performance as both proximal descents, and hence could not compete with those two algorithms.

\begin{figure}[!ht]
\scriptsize
\centering
\begin{subfigure}[c]{0.48\textwidth}
    \includegraphics[width=\textwidth]{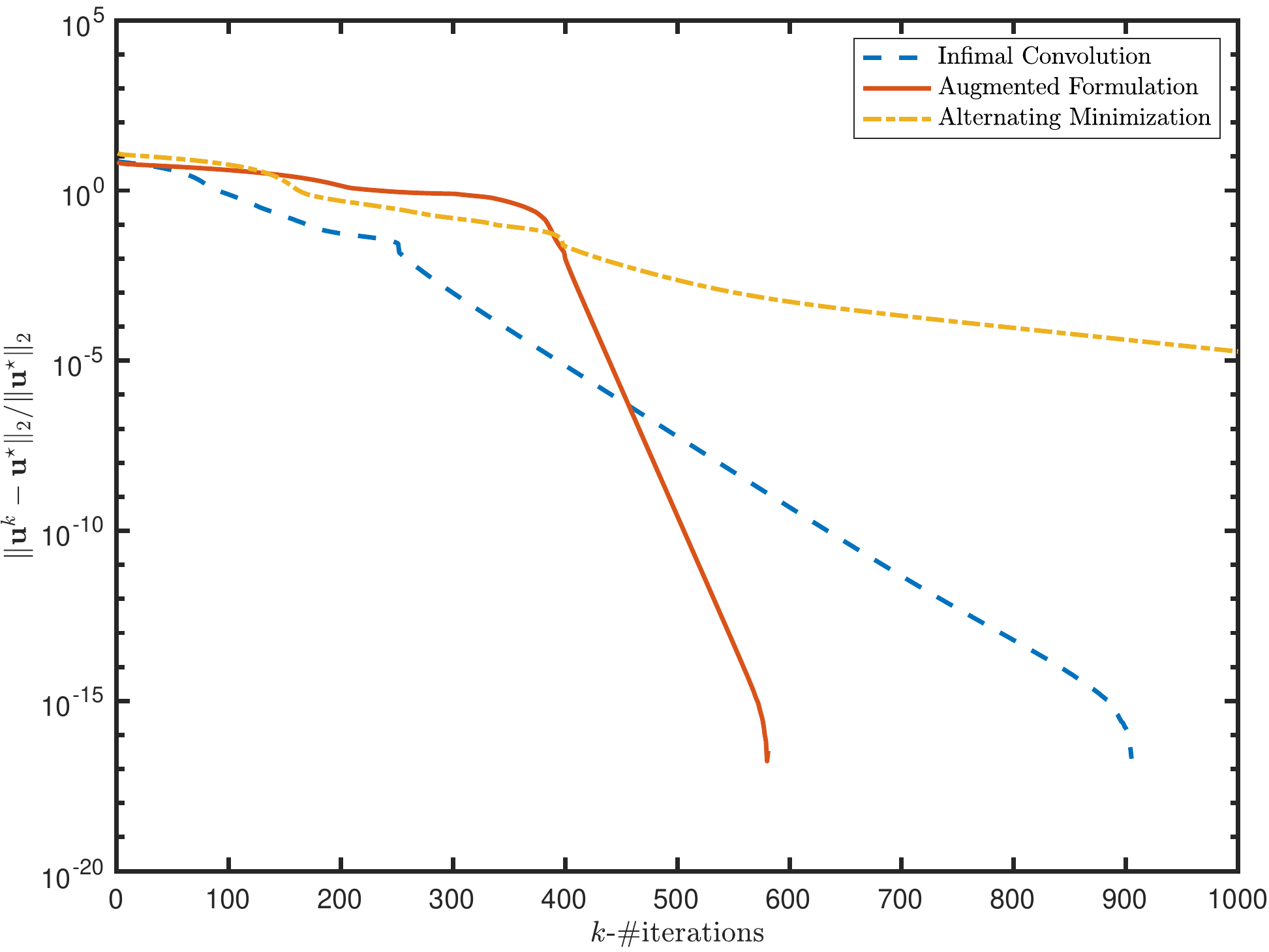}
    \caption{Convergence Rate}
    \label{fig:CCComparisonA}
\end{subfigure} 
\quad
\begin{subfigure}[c]{0.48\textwidth}
    \includegraphics[width=\textwidth]{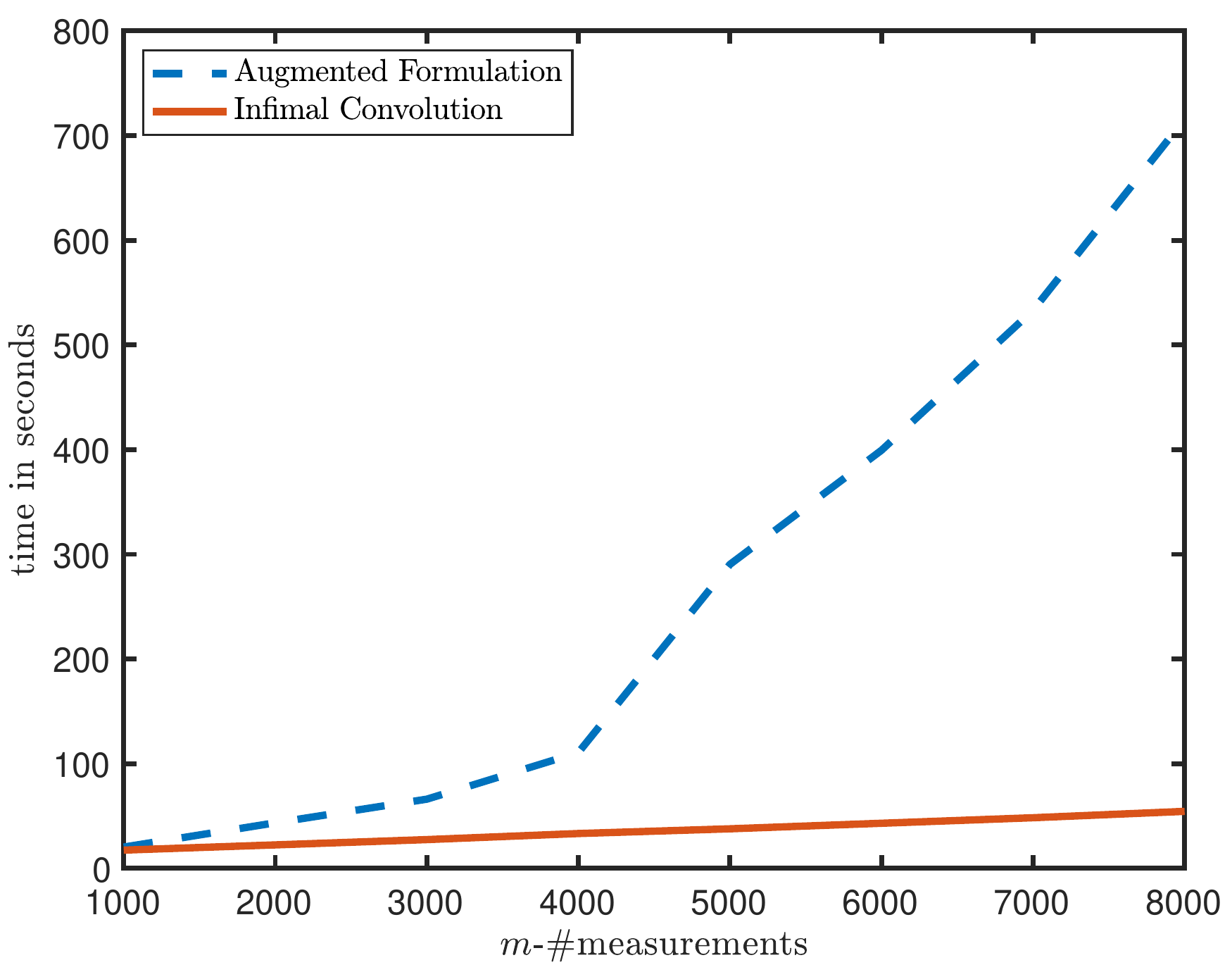}
    \caption{Running Time (50 iterations)}
    \label{fig:CCComparisonB}
\end{subfigure}
    \caption{In the left panel we look at the relative error
    with respect to the number of times the proximal operator is called for $\Amat\in\bbR^{100\times500}$ and $\Vu^\dagger \in \bbR^{500}$ is $14$-sparse. 
    In the right panel we compare average running time of augmented and infimal convolution formulations when reconstructing a $100$-sparse signal $\Vu^\dagger \in \bbR^{5000}$ from $m$ measurements, that vary from $1000$ to $8000$. 
    }
    \label{fig:CCCComparison}
\end{figure}


\section{Discussion}

In the present work we discussed the benefits of multi-penalty regularization for support recovery of signals when pre-measurement noise is amplified by the measurement operator and numerical challenges in solving the corresponding variational formulation. 
Since alternating minimization is for this task sub-optimal in terms of both the computational efficiency and theoretical analysis, we proposed a novel reduction to single-penalty regularization based on infimal convolution, and compared this new approach to an existing reduction based on augmented formulations.
Moreover, we established linear convergence for both single-penalty reductions and showed that our new approach omits a computational bottleneck that is unavoidable in the augmented approach, and causes a significant additional computational workload if the number of measurements increases.
There are several interesting open questions left for future work.

First, in Remark \ref{rem:MoreauInterpretation} we observed, for $q = 1$, a connection between the infimal convolution formulation and the proximal descent on the $\ell_2$-norm of the gradient of a Moreau-regularized $\ell_1$-functional. As we have not seen a comparable relation in the context of multi-penalty regularization so far, we are curious whether this observation can be extended to the case $0 < q < 1$. If so, this might provide valuable insights into non-convex optimization.

Second, as the reader might have noticed, great parts of the arguments we used (support stabilization, sign stabilization, etc.) are not restricted to finite dimensions. In light of more general settings of multi-penalty regularization in \cite{NP14} and single-penalty regularization in \cite{BL15}, it would be fruitful to transfer our findings to general separable Hilbert spaces as well.

Third, we mention that when using the infimal convolution based approach, in some experiments it was possible to choose $\mu$ much larger than suggested by Theorem \ref{thm:IC_Convergence}, while still observing reliable convergence of the program. 
We wonder whether there is an alternative proof leading to a relaxed condition on $\mu$ resembling the assumption in Theorem \ref{thm:Rho}.

Let us conclude by emphasizing that the infimal convolution formulation can as well be applied if regularizers other than the $\ell_q$-norm are used in the multi-penalty problem, e.g. Smoothly Clipped Absolute Deviation (SCAD) \cite{fan2001variable}, Minimax Concave Penalty (MCP) \cite{zhang2010nearly}, and Log-Sum Penalty (LSP) \cite{candes2008enhancing}. In those cases the more general single-penalty rate analysis in \cite{ZLX14} should prove useful as a tool.

\section*{Acknowledgment}

ZK and VN acknowledge the support from RCN-funded FunDaHD project No 251149/O70. JM acknowledges the support of DFG-SPP 1798.


\appendix

\section{Proofs}

\subsection{Proof of Lemma \ref{lem:mp_decoupled_sol}} \label{app:mp_decoupled_sol}
For a fixed $\Vu$ the minimization of $\CT_{\alpha,\beta}^q$ in \eqref{eqn:mp_problem} with respect to $\Vv$  reduces to Tikhonov minimization, and thus the solution satisfies
\begin{align} \label{eq:vu}
    \Vv = v(\Vu) = \LRP{\beta\,\Id_n + \Amat^\top\Amat}^{-1}\LRP{\Amat^\top \Vy-\Amat^\top\Amat \Vu}.
\end{align}
Rewriting the above expression we have
\[\beta v(\Vu) = \Amat^\top\LRP{\Vy-\Amat \Vu} - \Amat^\top\Amat v(\Vu) .\]
Plugging this expression into \eqref{eqn:mp_problem} the minimization problem for $u$ is rewritten as 
\[\CT_{\alpha,\beta}^q (\Vu,v(\Vu)) = \frac{1}{2}\EUSP{\Amat(\Vu+v(\Vu))-\Vy}{\Amat \Vu -\Vy} + \frac{\alpha}{q}\N{\Vu}_q^q.\]
The Woodbury identity for invertible matrices $\Vmat\in\bbR^{m\times m}$, $\Wmat\in\bbR^{n\times n}$ and matrices $\Mmat_1\in\bbR^{m\times n}$, $\Mmat_2\in\bbR^{n\times m}$ reads
\begin{equation} \label{eqn:Woodbury_ident}
\LRP{\Vmat\!+\!\Mmat_1\Wmat^{-1}\Mmat_2}^{-1}\!\! = \Vmat^{-1}\! - \Vmat^{-1}\Mmat_1\!\LRP{\Wmat+\Mmat_2\Vmat^{-1}\Mmat_1}^{-1}\Mmat_2\!\Vmat^{-1}.
\end{equation}
Using \eqref{eq:vu}, this gives
\begin{align*} \Amat(\Vu+v(\Vu))-\Vy &= \Amat v(\Vu) + \Amat \Vu -\Vy\\&= \LRP{\Id_m -\Amat\LRP{\beta\Id_n+\Amat^\top\Amat}^{-1}\Amat^\top}\LRP{\Amat \Vu -\Vy}\\&= \LRP{\Id_m +\frac{\Amat\Amat^\top}{\beta}}^{-1}\LRP{\Amat \Vu-\Vy}.\end{align*}
Plugging this expression back into $\CT_{\alpha,\beta}^q (\Vu,v(\Vu))$, and extracting the square root, we have $\CT_{\alpha,\beta}^q (\Vu,v(\Vu))= \CF_\beta(\Vu)$.
Minimizing over $\Vu$ and using the following simple observation gives the conclusion.

\begin{lemma} \label{lem:LocalMin}
    If $\uabq$ is a local minimizer of \eqref{eqn:q_augmented_problem}, then the pair $(\uabq,v(\uabq))$ with $v(u)$ defined in \eqref{eqn:vOFu}, is a local minimizer of $\CT_{\alpha,\beta}^q$ in \eqref{eqn:mp_problem}. 
\end{lemma}

\begin{proof}
    Let $\uabq$ be a local minimizer of $\eqref{eqn:q_augmented_problem}$ and assume there exists a sequence $(\Vu^k,\Vv^k) \rightarrow (\uabq,v(\uabq))$ such that $\Tabq(\Vu^k,\Vv^k) < \Tabq(\uabq,v(\uabq))$, for all $k \in \mathbb{N}$. 
    We then have
    \begin{align*}
        \Fb(\Vu^k) &= \Tabq(\Vu^k,v(\Vu^k)) \le \Tabq(\Vu^k,\Vv^k)< \Tabq(\uabq,v(\uabq)) = \Fb(\uabq),
    \end{align*}
    where the first inequality follows from the minimality of $v(\Vu^k)$.
    This contradicts the assumption that $\uabq$ is a local minimizer of $\eqref{eqn:q_augmented_problem}$.
\end{proof}

\subsection{Proof of Lemma \ref{lem:AM}} \label{app:AM}

First note that
\begin{align*}
    &\prox_{\mu, \frac{\alpha}{q} \Nq{\cdot}}(\Vu^k - \mu \Bb^\top(\Bb \Vu^k - \Vy_\beta)) = \\
    &\quad\prox_{\mu, \frac{\alpha}{q} \Nq{\cdot}} \Bigg( \Vu^k - \mu\Amat^\top \LRP{ \Id_m + \frac{\Amat\Amat^\top}{\beta} }^{-1} \round{\Amat \Vu^k - \Vy} \Bigg)
\end{align*}
while
\begin{align*}
    &\prox_{\mu, \frac{\alpha}{q} \Nq{\cdot}} (\Vu^k - \mu \Amat^\top(\Amat \Vu^k + \Amat v(\Vu^k) - \Vy)) = \\
    &\quad\prox_{\mu, \frac{\alpha}{q} \Nq{\cdot}} \Bigg( \Vu^k - \mu \round{\Amat^\top - \Amat^\top\Amat (\beta \Id_n + \Amat^\top\Amat)^{-1} \Amat^\top} \round{\Amat \Vu^k - \Vy} \Bigg).
\end{align*}
Hence, it suffices to show that
\begin{align*}
    \Amat^\top \left( \Id_m + \frac{\Amat\Amat^\top}{\beta} \right)^{-1} = \Amat^\top - \Amat^\top\Amat (\beta \Id_n + \Amat^\top\Amat)^{-1} \Amat^\top .
\end{align*}
Extracting $\Amat^\top$ from the left and using the Woodbury identity \eqref{eqn:Woodbury_ident} with $\Mmat_1=\Amat, \Mmat_2 = \Amat^\top$, $\Wmat=\beta\Id_n$, and $\Vmat=\Id_m$ the conclusion follows.

\subsection{Proof of Theorem \ref{thm:Rho}}
\label{app:Rho}

In order to prove Theorem \ref{thm:Rho}, we have to control the eigenvalues of $\Bb^\top\Bb$ characterizing the growth of the data fidelity term in \eqref{eqn:q_augmented_problem}.
\begin{lemma} \label{lem:Bb_bounds}
    For $\Bb \in \R^{m\times n}$ defined as in Lemma \ref{lem:mp_decoupled_sol},
    \begin{align*}
        L := \opN{\Bb^\top \Bb} = \LRP{\opN{\Amat}^{-2} + \beta^{-1}}^{-1},
    \end{align*}
    is the Lipschitz-constant of the gradient of the augmented data-fidelity term $\frac{1}{2}\Ntwo{\Bb \Vu - \Vy_\beta}^2$. Moreover, for any $I \subset [n]$,
    \begin{align*}
        \lmin(\BbI^\top\BbI) \ge \left( 1 + \frac{\opN{\Amat}^2}{\beta} \right)^{-1} \lmin(\Amat_I^\top \Amat_I).
    \end{align*}
\end{lemma}
\begin{proof}
    Let $\Amat = \Umat\VSigma \Vmat^\top$ denote the SVD of $\Amat$.
    This gives
    \begin{align}\label{eqn:BbTBb_SVS}
            \Bb^\top\Bb = \Vmat\VSigma^\top \left( \Id_m + \frac{\VSigma \VSigma^\top}{\beta} \right)^{-1} \VSigma\Vmat^\top,
    \end{align}
    so that $
        \opN{\Bb^\top\Bb}= \LRP{\opN{\Amat}^{-2} + \beta^{-1}}^{-1}.$
    By \eqref{eqn:BbTBb_SVS}, we have for any $\Vz \in \bbR^n$
    \begin{align}\label{eqn:Bb_Lmin}
        \SN{\Vz^\top \Bb^\top\Bb \Vz} 
        &= \SN{\Vz^\top \Vmat\VSigma^\top \left( \Id_m + \frac{\VSigma \VSigma^\top}{\beta} \right)^{-1} \VSigma \Vmat^\top \Vz} \nonumber
        \ge \left( 1 + \frac{\opN{\Amat}^2}{\beta} \right)^{-1} \SN{\Vz^\top \Vmat \VSigma^\top \VSigma \Vmat^\top \Vz}\nonumber \\
        &= \LRP{1 + \frac{\opN{\Amat}^2}{\beta} }^{-1} \SN{\Vz^\top \Amat^\top \Amat \Vz},
    \end{align}
    implying the second claim.
\end{proof}
We can now show that all, up to finitely many, iterates $\LRP{\Vu^k}_{k=1}^\infty$ generated by  \eqref{eqn:algorithm} share the same support and sign pattern. The proof is standard and follows \cite{BL15}. 
\begin{lemma}[Support and sign recovery] \label{lem:SupportRecovery}
    Assume $\beta > 0$, $0 < q \le 1$, and $\mu < \opN{\Amat}^{-2} + \beta^{-1}$. Then the iterates $\LRP{\Vu^k}_{k=1}^\infty$ satisfy $\Ntwo{\Vu^{k+1}-\Vu^k}\rightarrow 0$ as $k\rightarrow\infty$.
    Moreover, all iterates, up to finitely many, have the same support and sign pattern.
\end{lemma}
\begin{proof}
   Since $\mu< \opN{\Amat}^{-2} + \beta^{-1}= \frac{1}{L}$ we have $\Ntwo{\Vu^{k+1}-\Vu^k}\rightarrow 0$ as $k\rightarrow\infty$ by \cite[Corollary 2.1]{BL15}.
   Now, since the  range of $\prox_{\mu,\lambda\psi }$ is 
   $(-\infty,-\lambda_{\mu,q}] \cup \{0\}\cup [\lambda_{\mu,q},\infty)$, 
   it follows that the the absolute value of a non-zero entry of $\Vu^k$, for $k \ge 1$, is at least $\lambda_{\mu,q}$. 
   Thus, if $\supp(\Vu^{k+1})\neq \supp(\Vu^k)$ we have $\Ntwo{\Vu^{k+1}-\Vu^k}\geq \lambda_{\mu,q}$, and analogously, if $\Sgn(\Vu^{k+1})\neq \Sgn(\Vu^k)$ we have $\Ntwo{\Vu^{k+1}-\Vu^k}\geq 2\lambda_{\mu,q}$.
    Thus, since $\Ntwo{\Vu^{k+1}-\Vu^k}\rightarrow 0$ as $k\rightarrow\infty$, sign and support can change only finitely many times.
\end{proof}

\begin{proof}[Proof of Theorem \ref{thm:Rho}]
    By Lemma \ref{lem:SupportRecovery} there exists $k_0$ such that for all $k\geq k_0$ the support of $\Vu^k$ is finite, and support and sign of $\Vu^k$ is equal to that of $\Vu^\star$.
    Thus, by \cite[Proposition 2.3]{BL15}, $\Vu^\star$ is a fixed point of \eqref{eqn:algorithm}.
    Denote $I=\supp(\Vu^\star)$ with $|I| \le s$. 
    The definition of proximal operator in \eqref{eqn:proxOp} and the Karush-Kuhn-Tucker (KKT) conditions yield
    \[ \alpha\Sgn(\usf_i^\star) \SN{\usf^\star_i}^{q-1} = -(\Bb^\top(\Bb \Vu^\star - \Vy_\beta))_i, \quad i\in I, \]
    and
    \begin{align*} 
        \usf^{k+1}_i+&\alpha\mu\Sgn(\usf_i^{k+1})\big\lvert\usf_i^{k+1}\big\rvert^{q-1} =\usf_i^k-\mu(\Bb^\top(\Bb \Vu^{k} - \Vy_\beta))_i, \, \quad i\in I.
    \end{align*}
    Subtracting the two equations on the index set $I$, and denoting $\psi(\Vu) = \frac{1}{q} \N{\Vu}_q^q$, we have 
    \begin{align}\label{eqn:step} 
        \Vu_I^{k+1} - \Vu^\star_I &+\alpha\mu\LRP{\psi'(\Vu_I^{k+1}) - \psi'(\Vu_I^\star)} = \Vu_I^{k} - \Vu^\star_I  - \mu\LRP{\Bb^\top\Bb(\Vu^{k} - \Vu^\star)}_I,
    \end{align}
    where $\psi'(\Vu)=(\Sgn(\usf_i)\SN{\usf_i}^{q-1})_{i\in [n]}$ is acting entry-wise. 
    Note that since $k\geq k_0$ we have $\Sgn(\Vu^\star_I)=\Sgn(\Vu^{k+1}_I)$ and  $\N{\Vu^\star-\Vu^k}_2=\Ntwo{\Vu^\star_I-\Vu^k_I}$.
    A straightforward calculation gives
    \[ \Vu_I^{k} - \Vu^\star_I  - \mu\LRP{\Bb^\top\Bb(\Vu^{k} - \Vu^\star)}_I=(\Id_s-\mu\Mmat_{I,I})\LRP{\Vu_I^{k} - \Vu^\star_I}\]
    where $\Mmat=\Bb^\top\Bb$.
    Taking the inner product of \eqref{eqn:step} with $\Vu_I^{k+1} - \Vu^\star_I$, and applying the Cauchy-Schwartz inequality, we get
    \begin{align*}
        \Ntwo{\Vu_I^{k+1} - \Vu^\star_I}^2&-\alpha\mu\EUSP{\Vu_I^{k+1} - \Vu^\star_I}{\psi'(\Vu_I^{k+1}) - \psi'(\Vu^\star_I)} \leq \opN{\Id_s-\mu\Mmat_{I,I}} \Ntwo{\Vu_I^{k+1} - \Vu^\star_I}\Ntwo{\Vu_I^{k} - \Vu^\star_I}.
    \end{align*}
    Since $\psi$ is twice differentiable, and $\Vu^{k+1}$ and $\Vu^\star$ have the same sign and support, we have for the second term
    \begin{align*} 
        \Big\langle \Vu_I^{k+1} - \Vu^\star_I,&\psi'(\Vu_I^{k+1}) - \psi'(\Vu^\star_I)\Big\rangle =\sum_{i\in I} (\usf^{k+1}_{i}-\usf^\star_i)\LRP{\psi'(\usf^{k+1}_{i})-\psi'(\usf^\star_i)}= \sum_{i\in I} \psi''(C_{i}^{k+1})(\usf^{k+1}_{i}-\usf^\star_i)^2,
    \end{align*}
    where $C_i^{k+1}$ lies between $\usf^{k+1}_{i}$ and $\usf^\star_i$, and $\psi''(\usf) = (q-1) \usf^{q-2}$. Since $\Vu^{k} \rightarrow \Vu^\star$, we may assume $k_0$ sufficiently large to guarantee $\usf_i^k \ge \frac{1}{2} \usf_i^\star$, for all $k \ge k_0$ and $i \in I$. Consequently,
    \begin{align*} 
        \SN{\psi''(C_i^{k+1})} &= \SN{q-1} |C_i^{k+1}|^{q-2} \leq (1-q) \left( \frac{\umin}{2} \right)^{q-2}.
    \end{align*}
    Thus,
    \begin{align*}
        \Ntwo{\Vu_I^{k+1} - \Vu^\star_I}^2 &-\alpha\mu\EUSP{\Vu_I^{k+1} - \Vu^\star_I}{\psi'(\Vu_I^{k+1}) - \psi'(\Vu^\star_I)} \\&\geq \left( 1- \mu \alpha (1-q) \left( \frac{\umin}{2} \right)^{q-2} \right) \Ntwo{\Vu_I^{k+1} - \Vu^\star_I}^2.
    \end{align*}
    On the other hand, since $\mu < (\lambda_{\max}(\Mmat))^{-1} \le (\lambda_{\min}(\Mmat_{I,I}))^{-1}$, we have
    \begin{align*}
        \opN{\Id_s-\mu\Mmat_{I,I}} &= 1-\mu\lambda_{\min}(\Mmat_{I,I}) \leq 1-\mu\left( 1 + \frac{\opN{\Amat}^2}{\beta} \right)^{-1} \lmin(\Amat_I^\top \Amat_I),
    \end{align*}
    by Lemma \ref{lem:Bb_bounds}. Thus,
    \[ \Ntwo{\Vu^{k+1} - \Vu^\star}\leq\frac{1-\mu\left( 1 + \frac{\opN{\Amat}^2}{\beta} \right)^{-1} \lmin(\Amat_I^\top \Amat_I)}{ 1- \mu \alpha (1-q) \left( \frac{\umin}{2} \right)^{q-2} } \Ntwo{\Vu^{k} - \Vu^\star}.\]
    Together with the RIP of $\Amat$ this yields the claim. 
\end{proof}

\subsection{Proof of Lemma \ref{lem:Moreau_prox}}
\label{app:Moreau_prox}
Let $\Vx \in \R^n$ be fixed and assume $f(0) = 0$ without loss of generality. We have
\begin{align*}
    \prox_{\mu,\lambda M_{t,f}} (\Vx) 
    &= \argmin_{\Vz \in \R^n} \frac{1}{2} \Ntwo{\Vz - \Vx}^2 + \mu \lambda M_{t,f}(\Vz) \\
    &= \argmin_{\Vz \in \R^n} \inf_{\tilde{\Vz} \in \R^n} \frac{1}{2} \Ntwo{\Vz - \Vx}^2 + \mu \lambda f(\tilde{\Vz}) + \frac{\mu \lambda}{2t} \Ntwo{\Vz - \tilde{\Vz}}^2 \\
    &= \argmin_{\Vz \in \R^n} \inf_{\tilde{\Vz} \in \R^n} h(\Vz,\tilde{\Vz}).
\end{align*}
By $f$ being lower semi-continuous and bounded from below, we have
\begin{align*}
    \inf_{\Vz,\tilde{\Vz} \in \R^n} h(\Vz,\tilde{\Vz}) = \min_{\Vz,\tilde{\Vz} \in \R^n} h(\Vz,\tilde{\Vz}),
\end{align*}
implying $\prox_{\mu,\lambda M_{t,f}} (\Vx)\neq\emptyset$.
Denote by $\CE_{\tilde{\Vz}} = \curly{ \theta \Vx + (1-\theta) \tilde{\Vz} \colon \theta \in [0,1]}$ the line connecting $\Vx$ and $\tilde{\Vz}$. 
Since $\CE_{\tilde{\Vz}}$ is convex, we have $h(\VP_{\CE_{\tilde{\Vz}}}(\Vz),\tilde{\Vz}) \le h(\Vz,\tilde{\Vz})$, for any $\Vz,\tilde{\Vz} \in \R^n$, with equality if and only if $\VP_{\CE_{\tilde{\Vz}}}(\Vz) = \Vz$. 
Consequently, if $(\Vz,\tilde{\Vz})$ solves the above program, we have $\Vz = \theta \Vx + (1-\theta) \tilde{\Vz}$ for some $\theta \in [0,1]$. 
Let us define
\begin{align*}
    \tilde{h}(\theta,\tilde{\Vz}) = h( \theta \Vx + (1-\theta) \tilde{\Vz}, \tilde{\Vz}) = \round{ \frac{(1-\theta)^2}{2} + \frac{\mu \lambda \theta^2}{2t} } \Ntwo{\Vx - \tilde{\Vz} }^2 + \mu \lambda f(\tilde{\Vz}).
\end{align*}
By the above considerations we have
\begin{align*}
    \min_{\Vz,\tilde{\Vz} \in \R^n} h(\Vz,\tilde{\Vz}) = \min_{\tilde{\Vz} \in \R^n} \min_{\theta \in [0,1]} \tilde{h}(\theta,\tilde{\Vz}),
\end{align*}
where there is a one-to-one correspondence between solutions $(\Vz^\star,\tilde{\Vz}^\star)$ of the left side and solutions $(\theta^\star,\tilde{\Vz}^\star)$. 
Moreover, it follows easily that for $\tilde{\Vz} \in \R^n$ fixed,
\begin{align*}
    \theta^\star = \argmin_{\theta \in [0,1]} \tilde{h}(\theta,\tilde{\Vz}) = \frac{1}{1 + \frac{\mu \lambda}{t}},
\end{align*}
which is independent of $\tilde{\Vz}$. 
Thus, the claim follows since
\begin{align*}
    \argmin_{\tilde{\Vz} \in \R^n}\tilde{h}(\theta^\star,\tilde{\Vz}) 
    &= \argmin_{\tilde{\Vz} \in \R^n}\frac{\mu \lambda }{t}{\frac{1}{1 + \frac{\mu \lambda}{t}}} \frac{\Ntwo{ \Vx - \tilde{\Vz} }^2}{2} + \mu \lambda f(\tilde{\Vz}) =\prox_{(t+\mu\lambda), f}(\Vx).
\end{align*}

\subsection{Proof of Theorem \ref{thm:IC_Convergence}}
\label{app:IC_Convergence}

As in the proof of Theorem \ref{thm:Rho}, the first step is to control support and signs of the iterates. \BLUE{Recall that, for $\Vw^k$ as in \eqref{eqn:InfConv_algorithm}, we denote by $\Vu^k = \prox_{\frac{1}{\beta},\frac{\alpha}{q} \Nq{\cdot}} (\Vw^k)$ the sequence of minimizers attaining $g(\Vw^k)$, by $\Vv^k = \Vw^k - \Vu^k$, and that by \eqref{eqn:InfConv_AMrepresentation} we have
\begin{align} \label{eqn:InfConv_AMrepresentation2}
    \begin{cases}
    \Vw^k = \argmin_{\Vw \in \R^n} \frac{1}{2\mu} \Ntwo{\Vw - \Vw^{k-1} + \mu \Amat^\top (\Amat \Vw^{k-1} - \Vy)}^2 + \frac{\beta}{2} \Ntwo{\Vw - \Vu^k}^2 \\
    \Vu^k = \argmin_{\Vu \in \R^n} \frac{\beta}{2} \Ntwo{\Vu - \Vw^k}^2 +  \frac{\alpha}{q} \Nq{\Vu}
    \end{cases}.
\end{align}{}
}   
\begin{lemma}[Sign and support stability] \label{lem:InfConv_SignSupport}
    Assume $\mu < \opN{\Amat}^{-2}$. Then the successive iterates $\Ntwo{\Vw^{k+1} - \Vw^k}$, $\Ntwo{\Vu^{k+1} - \Vu^k}$, and $\Ntwo{\Vv^{k+1} - \Vv^k}$ converge to zero and all but finitely many iterates $\Vu^k$ share the same finite support and the same signs.
\end{lemma}{}
\begin{proof}
   First, note that $g$ is a proper and coercive function. Second, as $g(\Vw) = \inf_{\Vu \in \R^n} f(\Vu,\Vw)$, for $f$ continuous, we obtain continuity of $g$ at any point $\Vw \in \R^n$ since by coercivity of $f$ the infimum can be restricted to a finite ball and the infimum of continuous functions on a compact set is continuous.
   Consequently, by \cite[Corollary 2.1]{BL15} and the assumption on $\mu$ we have $\Ntwo{\Vw^{k+1} - \Vw^k} \rightarrow 0$, for $\Vw^{k+1} = \prox_{\mu,g} ( \Vw^k - \mu \Amat^\top (\Amat \Vw^k - \Vy))$. 
   By the KKT-conditions of \eqref{eqn:InfConv_AMrepresentation2}, we obtain
   \begin{align*}
       0 &= (\Vw^{k+1} - \Vw^k) + \mu \Amat^\top (\Amat \Vw^k - \Vy) + \beta \mu \Vv^{k+1}, \\
       0 &= (\Vw^k - \Vw^{k-1}) + \mu \Amat^\top (\Amat \Vw^{k-1} - \Vy) + \beta \mu \Vv^k.
   \end{align*}
   Subtracting the two equations gives $\Ntwo{\Vv^{k+1} - \Vv^k} \rightarrow 0$, and $\Vu^k = \Vw^k - \Vv^k$ yields $\Ntwo{\Vu^{k+1} - \Vu^k} \rightarrow 0$. 
   The second claim follows as in Lemma \ref{lem:SupportRecovery}, since $\Vu^k$ is a thresholded version of $\Vw^k$.
\end{proof}{}
\begin{proof}[Proof of Theorem \ref{thm:IC_Convergence}]
    First note that $\Vw^k \rightarrow \Vw^\star$ implies via Lemma \ref{lem:InfConv_SignSupport} that $\Vu^k \rightarrow \Vu^\star$ and $\Vv^k \rightarrow \Vv^\star$. Furthermore, $\Vw^\star$ is a fixed point of \eqref{eqn:InfConv_algorithm}, by \cite[Proposition 2.3]{BL15}.
    By Lemma \ref{lem:InfConv_SignSupport} there exists $k_0$ such that for all $k\geq k_0$ the support of $\Vu^k$ is finite, and support and sign of $\Vu^k$ is equal to that of $\Vu^\star$.
    Denote $I=\supp(\Vu^\star)$. 
    By the KKT-conditions of \eqref{eqn:InfConv_AMrepresentation2}, we get
    \begin{align*}
        i \in I \colon \begin{cases}{}
        \alpha \mu \sign(\usf_i^\star) |\usf_i^\star|^{q-1}= -\mu (\Amat^\top ( \Amat \Vw^\star - \Vy ))_i, \\
        \wsf_i^{k+1} + \alpha \mu \sign(\usf_i^{k+1}) |\usf_i^{k+1}|^{q-1} = \wsf_i^k -\mu (\Amat^\top ( \Amat \Vw^k - \Vy ))_i,
        \end{cases}
    \end{align*}{}
    and
    \begin{align*}
        i \notin I \colon \begin{cases}{}
        0 = \beta \mu \wsf_i^\star + \mu (\Amat^\top ( \Amat \Vw^\star - \Vy ))_i, \\
        0 = (1 + \beta\mu) \wsf_i^{k+1} - \wsf_i^k + \mu (\Amat^\top ( \Amat \Vw^k - \Vy ))_i,
        \end{cases}.
    \end{align*}{}
    For $\psi(\Vu) = \frac{1}{q} \N{\Vu}_q^q$ with $\psi'(\Vu)=(\Sgn(\usf_i)\SN{\usf_i}^{q-1})_{i\in [n]}$ acting entry-wise, this implies
    \begin{align} \label{eqn:InfConv_KKT_consequence1}
        (\Vw^{k+1} - \Vw^\star)_I + \alpha\mu ( \psi'(\Vu^{k+1}) - \psi'(\Vu^\star) ) = (\Vw^k - \Vw^\star)_I - \mu \Amat_I^\top \Amat (\Vw^k - \Vw^\star)
    \end{align}{}
    and
    \begin{align} \label{eqn:InfConv_KKT_consequence2}
        (1+\mu\beta) (\Vw^{k+1} - \Vw^\star)_{I^c} = (\Vw^k - \Vw^\star)_{I^c} - \mu \Amat_{I^c}^\top \Amat (\Vw^k - \Vw^\star).
    \end{align}{}
    Repeating the steps as in Theorem \ref{thm:Rho}, from \eqref{eqn:InfConv_KKT_consequence1} we get
    \begin{align*}
        \round{ 1 - \alpha\mu (1-q) \round{\frac{\umin}{2}}^{q-2} }
        \Ntwo{(\Vw^{k+1} - \Vw^\star)_I}^2
        \le
        \opN{ \VP_I - \mu \Amat_I^\top \Amat } \Ntwo{\Vw^k - \Vw^\star} \Ntwo{(\Vw^{k+1} - \Vw^\star)_I}
    \end{align*}{}
    and from \eqref{eqn:InfConv_KKT_consequence2} we obtain
    \begin{align*}
        (1 + \mu\beta) \Ntwo{(\Vw^{k+1} - \Vw^\star)_{I^c}}
        \le
        \opN{ \VP_{I^c} - \mu \Amat_{I^c}^\top \Amat } \Ntwo{\Vw^k - \Vw^\star}.
    \end{align*}
Squaring and summing the last two equations, the claim follows by orthogonality of $(\Vw^{k+1} - \Vw^\star)_I$ and $(\Vw^{k+1} - \Vw^\star)_{I^c}$.
\end{proof}{}

\section{Coherence Bound}
The following Lemma bounds the coherence of $\Bb$ in terms of the coherence of $\Amat$. The bound becomes tight for large choices of $\beta$.
\begin{lemma} \label{lem:BbCoherence}
We have 
\begin{align*}
    \mathrm{coh}(\VB_\beta)\leq \left( 1+\frac{\N{\VA}^2}{\beta} \right) \mathrm{coh}(\VA) +\frac{\N{\VA}^2}{\beta}.
\end{align*}
\end{lemma}
\begin{proof}
Recall that the coherence of a matrix is defined as
\[ 
    \mathrm{coh}(\VM) = \max_{i\neq j}\frac{\SN{\Vm_i^\top\Vm_j}}{\Ntwo{\Vm_i}\Ntwo{\Vm_j}},
\]
where $\Vm_i$ is the $i$-th column of $\VM$.
Define $\VQ_\beta = \Id_m + \frac{\VA\VA^\top}{\beta}$, so that $\VB_\beta = \VQ_\beta^{-1/2}\VA$, and let $\VA=\VU\VSigma \VV^\top$ be the SVD of $\VA$.
This gives 
\[
    \VQ_\beta^{-1}-\Id_m = \Big(\Id_m + \frac{\VA\VA^\top}{\beta}\Big)^{-1}-\Id_m= \VU\bigg(\Big(\Id_m+\frac{\VSigma \VSigma^\top}{\beta}\Big)^{-1}-\Id_m \bigg)\VU^\top.
\]
Therefore,
\[ 
    \N{\VQ_\beta^{-1}-\Id_m} = \N{\Big(\Id_m+\frac{\VSigma \VSigma^\top}{\beta}\Big)^{-1}-\Id_m }
    = \frac{c_\beta}{1+c_\beta},
\]
for $c_\beta = \frac{\N{\VA}^2}{\beta}$, and by triangle inequality and Cauchy-Schwarz
\begin{align*}
    | \Vb_i^\top \Vb_j | 
    = | \Va_i^\top \VQ_\beta^{-1} \Va_j | 
    \le | \Va_i^\top \Va_j | + \frac{c_\beta}{1+c_\beta}  \Ntwo{\Va_i} \Ntwo{\Va_j},
\end{align*}
for all columns $\Vb_i,\Vb_j$ of $\Bb$. By the same argument we compute
\[
    \VQ_\beta^{-1/2}-\Id_m =\VU\bigg(\Big(\Id_m+\frac{\VSigma \VSigma^\top}{\beta}\Big)^{-1/2}-\Id_m \bigg)\VU^\top,
\]
giving
\[ 
    \opN{\VQ_\beta^{-1/2}-\Id_m} =
    1 - \sqrt{\frac{\beta}{\N{\VA}^2+\beta}}=1-({c_\beta+1})^{-1/2}.
\]
This yields
\begin{align} \label{eq:LowerBound}
    \Ntwo{\Vb_i} \geq \Ntwo{\Va_i}- \Ntwo{(\VQ_\beta^{-1/2}-\Id_m)\Va_i}\geq({c_\beta+1})^{-1/2}\Ntwo{\Va_i}
\end{align}
which implies
\begin{align*}
    \mathrm{coh}(\VB_\beta)
    = \max_{i\neq j} \frac{\SN{\Vb_i^\top\Vb_j}}{\Ntwo{\Vb_i}\Ntwo{\Vb_j}} 
    \leq (1+c_\beta) \left( \mathrm{coh}(\VA) + \frac{c_\beta}{1+c_\beta} \right).
\end{align*}
\end{proof}

For small $\beta$, the bound in Lemma \ref{lem:BbCoherence} is lossy. However, we can show that the coherence of $\Bb$ converges to the coherence of a conditioned version of $\Amat$, for $\beta \rightarrow 0$.

\begin{lemma} \label{lem:BbCoherence2}
  Let $\Amat \in \R^{m\times n}$, for $m \le n$, have full rank. We have $\mathrm{coh}(\Bb) \rightarrow \mathrm{coh}((\Amat\Amat^\top)^{-\frac{1}{2}} \Amat)$, for $\beta \rightarrow 0$. 
\end{lemma}
\begin{proof}
  Define $\VQ_\beta = \Id_m + \frac{\VA\VA^\top}{\beta}$, so that $\VB_\beta = \VQ_\beta^{-1/2}\VA$, and let $\VA=\VU\VSigma \VV^\top$ be the SVD of $\VA$. Define $\VC = \sqrt{\beta} (\Amat\Amat^\top)^{-\frac{1}{2}} \Amat$ with columns $\Vc_i$.
  First note, that
  \begin{align*}
      \left\| \VQ_\beta^{-1} - \round{\frac{\Amat\Amat^\top}{\beta}}^{-1} \right\|
      = \max_{i \in [m]} \abs{ \frac{1}{1+\frac{\sigma_i^2}{\beta}} - \frac{1}{\frac{\sigma_i^2}{\beta}} }
      = \abs{ \frac{\beta}{\sigma_{\text{min}}^2 \round{1+\frac{\sigma_{\text{min}}^2}{\beta}}} } 
      \le \frac{\beta^2}{\sigma_{\text{min}}^4}
  \end{align*}
  and
  \begin{align*}
      \left\| \VQ_\beta^{-\frac{1}{2}} - \round{\frac{\Amat\Amat^\top}{\beta}}^{-\frac{1}{2}} \right\|
      = \max_{i \in [m]} \abs{ \frac{1}{\sqrt{1+\frac{\sigma_i^2}{\beta}}} - \frac{1}{\sqrt{\frac{\sigma_i^2}{\beta}}} }
      = \sqrt{\beta} \abs{ \frac{\sqrt{\sigma_i^2} - \sqrt{\beta+\sigma_i^2} }{\sigma_{\text{min}} \sqrt{\beta+ \sigma_{\text{min}}^2}} } 
      \le \frac{\beta}{\sigma_{\text{min}}^2}.
  \end{align*}
  Consequently,
  \begin{align*}
      \abs{\inner{\Vb_i,\Vb_j} - \inner{\Vc_i,\Vc_j}} = \abs{\Ve_i^\top \Amat^\top \round{\VQ_\beta^{-1} - \round{\frac{\Amat\Amat^\top}{\beta}}^{-1}} \Amat \Ve_j } \le \beta^2 \frac{\opN{\Amat}^2}{\sigma_\text{min}^4}
  \end{align*}
  and
  \begin{align*}
      \abs{\Ntwo{\Vb_i} - \Ntwo{\Vc_i}} \le \Ntwo{\Vb_i - \Vc_i} = \left\| \round{\VQ_\beta^{-\frac{1}{2}} - \round{\frac{\Amat\Amat^\top}{\beta}}^{-\frac{1}{2}}} \Amat \Ve_i \right\| \le \beta \frac{\opN{\Amat}}{\sigma_\text{min}^2}.
  \end{align*}
  Since we have in addition that $\Ntwo{\Vc_i} \le \sqrt{\beta} \opN{(\Amat\Amat^\top)^{-\frac{1}{2}}\Amat} = \sqrt{\beta}$, $\Ntwo{\Vb_i} \le \opN{\VQ_\beta^{-\frac{1}{2}} \Amat} \le \sqrt{\beta}$, and $\Ntwo{\Vb_i} \ge (\opN{\Amat}^2 + \beta)^{-\frac{1}{2}} \Ntwo{\Va_i} \sqrt{\beta}$, we get
  \begin{align*}
      \abs{\frac{\inner{\Vb_i,\Vb_j}}{\Ntwo{\Vb_i}\Ntwo{\Vb_j}} - \frac{\inner{\Vc_i,\Vc_j}}{\Ntwo{\Vc_i}\Ntwo{\Vc_j}}}
      &= \abs{\frac{ \round{\inner{\Vb_i,\Vb_j} - \inner{\Vc_i,\Vc_j}} \Ntwo{\Vc_i}\Ntwo{\Vc_j} + \inner{\Vc_i,\Vc_j} \round{ \Ntwo{\Vc_i}\Ntwo{\Vc_j} - \Ntwo{\Vb_i}\Ntwo{\Vb_j} } }{\Ntwo{\Vb_i}\Ntwo{\Vb_j}\Ntwo{\Vc_i}\Ntwo{\Vc_j}}} \\
      &\le \frac{ \abs{ \inner{\Vb_i,\Vb_j} - \inner{\Vc_i,\Vc_j} } }{ \Ntwo{\Vb_i}\Ntwo{\Vb_j} } + \frac{\Ntwo{\Vc_i} \abs{ \Ntwo{\Vc_j} - \Ntwo{\Vb_j} } + \abs{ \Ntwo{\Vc_i} - \Ntwo{\Vb_i} } \Ntwo{\Vb_j} }{\Ntwo{\Vb_i}\Ntwo{\Vb_j}}\\
      &= \mathcal{O}(\beta) + \mathcal{O}(\sqrt{\beta}).
  \end{align*}
  We conclude by noting that $\mathrm{coh}(\VC) = \mathrm{coh}((\Amat\Amat^\top)^{-\frac{1}{2}} \Amat)$.
\end{proof}

\bibliographystyle{abbrv}
\bibliography{biblio.bib}

\end{document}